\documentclass[letterpaper,10pt]{IEEEtran}
\usepackage{graphicx,amsmath,amssymb,arydshln,color,cite,bm}
\usepackage[colorlinks=true]{hyperref}
\hypersetup{urlcolor=blue,linkcolor=blue,citecolor=blue,colorlinks=true}

\IEEEoverridecommandlockouts 

	\newtheorem{theorem}{Theorem}[section]
	\newtheorem{corollary}[theorem]{Corollary}
	\newtheorem{lemma}[theorem]{Lemma}
	\newtheorem{definition}{Definition}
	\newtheorem{remark}{Remark}

	\newcommand{\gain}{\gamma}
	\newcommand{\Gain}{\Gamma}
	
	\newcommand{\lb}{\left(}
	\newcommand{\rb}{\right)}
	
	\newcommand{\emax}[1]{\lambda_{\text{max}}\left( #1 \right)}
	\newcommand{\emin}[1]{\lambda_{\text{min}}\left( #1 \right)}

	\newcommand{\ub}{{\mathbf u}}
	\newcommand{\Ub}{{\mathbf U}}
	\newcommand{\Ubb}{{\bar{\mathbf U}}}
	\newcommand{\ubb}{{\bar{\mathbf u}}}
	\newcommand{\yb}{{\mathbf y}}
	\newcommand{\Yb}{{\mathbf Y}}
	\newcommand{\Ybb}{\bar{\mathbf Y}}
	\newcommand{\rbo}{{\mathbf r}}
	\newcommand{\Rb}{{\mathbf R}}
	\newcommand{\Rbb}{\bar{\mathbf R}}
	\newcommand{\vb}{{\mathbf v}}
	\newcommand{\Vb}{{\mathbf V}}
	
	\newcommand{\db}{{\mathbf d}}
	\newcommand{\Db}{{\mathbf D}}
	
	\newcommand{\wb}{{\mathbf w}}
	\newcommand{\Wb}{{\mathbf W}}
	\newcommand{\Wbb}{\bar{\mathbf W}}
	\newcommand{\zb}{{\mathbf z}}
	\newcommand{\Zb}{{\mathbf Z}}
	\newcommand{\Zbb}{\bar{\mathbf Z}}
	\newcommand{\Xb}{{\mathbf X}}
	\newcommand{\Xbb}{\bar{\mathbf X}}
	\newcommand{\Pb}{{\mathbf P}}
	\newcommand{\Gainb}{{\mathbf \Gain}}
	\newcommand{\Mb}{{\mathbf M}}

	\newcommand{\Mbb}{\bar{\mathbf M}}

	\newcommand{\Ubwh}{\widehat{\Ub}}

    \newcommand{\bfo}{{\bf1}}

    \newcommand{\covl}[2]{#1 \! \left\{ #2 \right\}}

	\newcommand{\matbegin}{\left[ }
	\newcommand{\matend}{\right] }

	\newcommand{\thbth}[9]{
	 	\begin{bmatrix}
	        	#1 & #2 & #3 \\
	                #4 & #5 & #6 \\
	        	#7 & #8 & #9
		\end{bmatrix}}

	\newcommand{\mattightbegin}{
	    \renewcommand{\baselinestretch}{1}
	    \renewcommand{\arraystretch}{.5}
	    \setlength{\arraycolsep}{0.15em}
	    \left[
	}
	\newcommand{\mattightend}{
	    \right]
	}
		
	\newcommand{\thbthtight}[9]{
	 \mattightbegin \begin{array}{ccc}
	        #1 & #2 & #3 \\
	                #4 & #5 & #6 \\
	        #7 & #8 & #9
	                \end{array}\mattightend}

	\newcommand{\R}{{\mathbb R}}
	
	\newcommand{\Z}{{\mathbb Z}}
	\newcommand{\cG}{{\cal G}}
	\newcommand{\cM}{{\cal M}}
	\newcommand{\cS}{{\cal S}}

	\newcommand{\cL}{{\mathbb L}}
	
	\newcommand{\gainb}{{\bm \gain}}

	    \newcommand{\expec}[1]{ {\mathbb E} \left[ #1 \right] }
	    \newcommand{\req}[1]{(\ref{#1.eq})}
	    
	    \newcommand{\be}{\begin{equation}}
	    \newcommand{\ee}{\end{equation}}
	    \newcommand{\beas}{\begin{eqnarray*}}
	    \newcommand{\eeas}{\end{eqnarray*}}
	    \newcommand{\bea}{\begin{eqnarray}}
	    \newcommand{\eea}{\end{eqnarray}}

	   \newcommand{\diag}[1]{{\rm{diag}} \left( #1 \right)}
	   \newcommand{\Diag}[1]{{\rm{Diag}} \left( #1 \right)}
	   
	   \newcommand{\tr}[1]{{\rm{tr}}\left( #1 \right)}

	   \newcommand{\vect}[1]{\mathrm{vec}\left( #1 \right)}

	   \newcommand{\bbm}{\begin{bmatrix}}
	   \newcommand{\ebm}{\end{bmatrix}}

\begin{document}

\title{An Input-Output Approach to \\ Structured Stochastic Uncertainty}


\author{Bassam~Bamieh,~\IEEEmembership{Fellow,~IEEE}, and Maurice Filo, \IEEEmembership{Member,~IEEE}%
				\thanks{
				B. Bamieh and M. Filo are with the department of Mechanical Engineering, the University of California at 
				Santa Barbara (UCSB). 
				
				This work is supported by NSF Awards ECCS-1408442.
				}
				}

\maketitle


\begin{abstract}
We consider linear time invariant systems with exogenous stochastic disturbances, and in feedback with structured stochastic uncertainties.  This setting encompasses linear systems with both additive and multiplicative noise. Our concern is to characterize second-order properties such as mean-square stability and performance.  A purely input-output treatment of these systems is given without recourse to state space models, and thus the results are applicable to certain classes of distributed systems. We derive necessary and sufficient conditions for mean-square stability in terms of the spectral radius of a linear matrix operator whose dimension is that of the number of uncertainties, rather than the dimension of any underlying state space models. Our condition is applicable to the case of correlated uncertainties, and reproduces earlier results for uncorrelated uncertainties. For cases where state space realizations are given, Linear Matrix Inequality (LMI) equivalents of the input-output conditions are given. 
\end{abstract}

\section{Introduction}
	Linear Time Invariant (LTI) systems driven by second order stochastic processes is a widely-used and powerful methodology for modeling and control of many physical systems in the presence of stochastic uncertainty. In the most well-known models, stochastic uncertainty enters the model equations additively. Linear systems with both additive and multiplicative stochastic signals are on the other hand  relatively less studied. This problem setting is important in the study of system robustness. While additive disturbances can represent uncertain forcing or measurement noise in a system, multiplicative disturbances are necessary to model uncertainty in system parameters and coefficients. When the multiplicative uncertainty is of the non-stochastic set-valued type, then the problem setting is the standard deterministic one of robust control~\cite{zhou1996robust}. The present paper is concerned with the stochastic multiplicative uncertainty setting, but the approach will appear to be closer to that of robust control compared to common stochastic treatments. 

Before commenting  on the background for the present work, a brief statement of the problem is given to allow for a more precise discussion. 	
	\begin{figure}[h]
		\begin{center}	
			\includegraphics[width=0.25\textwidth]{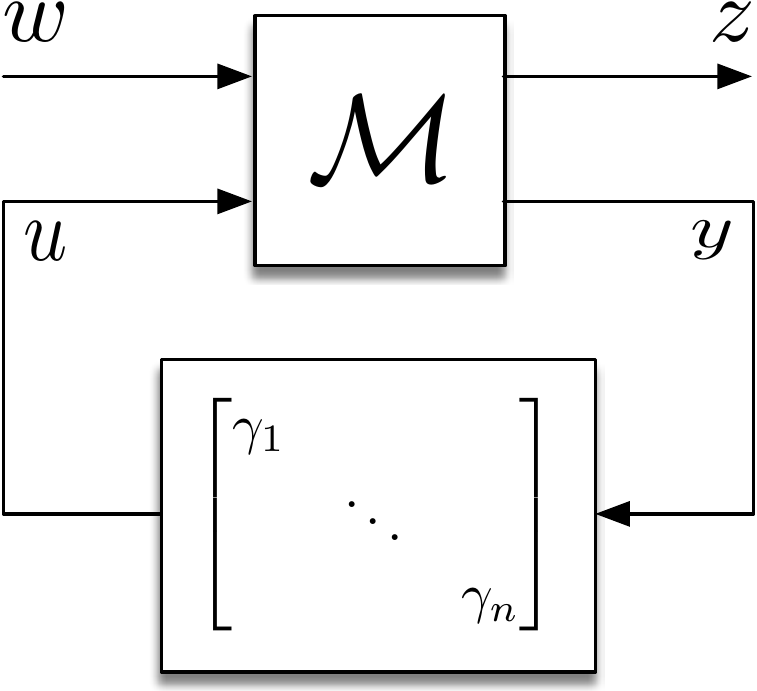}
		\end{center}
		\caption{The general setting of linear systems with both additive and multiplicative 
		stochastic disturbances. $\cM$ is a Linear Time Invariant (LTI) system, $w$ is a stationary stochastic 
		process that enters additively, while the multiplicative disturbances are modeled as a feedback through 
		 time-varying stochastic gains $\gain_1,\ldots,\gain_n$, 
		 represented here as a diagonal matrix acting on vector-valued
		 internal signals $u$ and $y$. The  signal $z$ represents an output whose variance quantifies a
		 performance measure.   }
	   \label{setting.fig}
	\end{figure}
	Figure~\ref{setting.fig} illustrates the setting  considered in this paper.  
	A Linear Time Invariant (LTI) 
	system $\cM$ is   in feedback with time-varying gains 
	$\gain_1$, $\ldots$, $\gain_n$. 
	These gains are random processes that are temporally independent, 
	but possibly mutually correlated. Another set of  stochastic disturbances are represented by the 
	vector-valued signal $w$ which enters additively, while the signal $z$ is an output whose variance 
	quantifies a performance measure. The feedback term is then a diagonal matrix with the individual 
	gains $\{\gain_i\}$ appearing on the diagonal. Such gains are commonly referred to as structured 
	uncertainties. 
	
	We should note the other common and related models in the literature which are usually formulated 
	in a state space setting. One such model is a linear system with a random ``A matrix'' such as 
	\be
		x(t+1) ~=~ A(t) x(t) ~+~ B w(t) , 
		\label{randA.eq}
	\ee
	where $A$ is a matrix-valued random process (with $A(t)$  independent of $\{x(\tau),\tau\leq t\}$). 
	Sometimes~\req{randA} can be rewritten in an alternative form using scalar-valued random 
	processes $\gain_1,\ldots,\gain_n$ ($n$ is not necessarily the dimension of $x$) 
	as follows 
	\be
		x(t+1) ~=~ \big( A_o+A_1\gain_1(t)+\cdots+A_n\gain_n(t)  \big)  x(t) ~+~ B w(t) .
		\label{randA_rewrite.eq}
	\ee
	Now this form can always be converted~\cite{zhou1996robust} to that of Figure~\ref{setting.fig}. The simplest case is when 
	the matrices $A_1,\ldots,A_n$ are all of rank 1, then each $\gain_i(t)$ in Figure~\ref{setting.fig} is a scalar 
	block, while otherwise, one would have  so-called repeated blocks. We refer the 
	reader to~\cite{zhou1996robust,packard1993complex} for this standard construction. 
	
%

The literature on systems with multiplicative noise goes back several decades. Early work considered models like~\req{randA},~\req{randA_rewrite}, but primarily in continuous time and using an Ito formulation. The primary tool~\cite{willems1976feedback} was to derive differential equations that govern the evolution of second order moments when the multiplicative noise is white, and  conditions for the asymptotic convergence of those deterministic equations are given in terms of solvability of certain Riccati-like equations. The case of colored multiplicative noise is less tractable since equations for the second order moments are not easy to obtain, although certain special classes have been studied~\cite{willems1975stability,blankenship1977stability}. For the white noise case however, more detailed analysis appeared in later work~\cite{boyd1994linear,el1995state}  which recast Mean Square Stability (MSS) conditions in terms of Linear Matrix Inequalities (LMIs). 
	
Another trend~\cite{elia2005remote,hinrichsen1995stability,el1992stability,lu2002mean} appeared in the 90's when systems with several sources of multiplicative noise were viewed in a similar manner to that of structured uncertainty common in robust control (as in the setting of Figure~\ref{setting.fig}). The interesting observation was made that MSS conditions can be given in terms of the spectral radius of a non-negative matrix of $H^2$ norms (the matrix of $H^2$ norms of the individual SISO subsystems of the MIMO  LTI system $\cM$). This criterion is analogous to necessary and sufficient conditions for robust stability to deterministic structured  time-varying uncertainty in both the $L^2$ and $L^\infty$-induced norms~\cite{dahleh1988necessary,megretski1993necessary,bamieh1993robust,doyle1982performance,khammash1991performance,shamma1994robust,fan1991robustness,doyle1985structured} settings. 

There are two  observations about the existing results~\cite{elia2005remote,hinrichsen1995stability,el1992stability,lu2002mean} that motivate the current work. The first is that although the final MSS conditions are stated in terms of input-output properties ($H^2$ norms), the arguments and proofs rely on state space realizations, LMIs related to those realizations, and scalings of those LMIs. Second, the existing results are for  multiplicative uncertainties that are mutually uncorrelated, and it is unclear how these arguments can be generalized to the correlated uncertainties case. It should be noted that the latter case is important for several applications. 

The aim of this paper is to provide a relatively elementary, and purely input-output treatment and derivation of the necessary and sufficient conditions for MSS and performance. In the process,  conditions  for the mutually correlated uncertainties case become transparent, as well as how special the uncorrelated case is. A new object is uncovered which can be termed the ``loop gain operator'', which acts on   covariances of signals in the feedback loop. We briefly describe this operator as a preview of the main result of this paper (the following statement is concerned only with MSS rather than performance, so the signal $w$ in Figure~\ref{setting.fig} is set to zero). Let $\Gainb$ be the mutual correlation matrix of the $\gain$'s, i.e. the matrix with $ij$'th entry 
$\Gainb_{ij} :=  \expec{ \gain_i(t) \gain_j^*(t)}$. Let $\{M_{22}(t)\}$ be the impulse response matrix sequence of the $\cM_{22}$ subsystem in Figure~\ref{setting.fig}, and define the {\em matrix-valued}  linear operator
\[
	\cL(X) ~:=~  \Gainb \circ  \left( 	\sum_{t=0}^\infty M_{22}(t) ~X~ M^*_{22}(t) \right), 
\]
where $\circ$ is the Hadamard (element-by-element) product of matrices. Note that it operates on matrices $X$ whose dimensions are the number of uncertainties, and not of any underlying state space realization. 
The operator $\cL$  is called the {\em loop gain operator} because it captures what happens to the covariance matrix of a signal as one goes ``once around the feedback loop'' in Figure~\ref{setting.fig} in the statistical steady state. 
The eigenvalues and ``eigen-matrices'' of this operator characterize MSS as well as the fastest growing second order statistics of the signals in the loop when mean square stability is lost. 
An examination of this operator shows that 
 in the more general setting $\Gainb\neq I$,  MSS conditions require not only calculations of $H^2$ norms, but also inner products of the various subsystems' impulse responses (of which the $H^2$ norms of subsystems are a special case). 
 The operator $\cL$ has several nice properties including monotonicity (preserving the semi-definite ordering on matrices), and consequently a Perron-Frobenius theory for its spectral radius and the associated eigen-matrix. These properties are described and exploited in the sequel.

	This paper is organized as follows. Section~\ref{prelims.sec} establishes   preliminary 
	results that are needed for the subsequent structured uncertainty analysis. One sided random
	 processes and their associated covariance sequences are defined. In addition, we provide 
	a natural input-output definition of Mean-Square Stability (MSS) in feedback systems. 
	The main tool we use is to convert the stochastic feedback system of Figure~\ref{setting.fig}
	to a deterministic feedback system that operates on matrix-valued signals, namely the covariance 
	sequences of all the signals in the feedback loop (Figure~\ref{M_Delta_cov.fig} provides a cartoon of
	this). LTI systems that operate on convariance-matrix-valued signals have some nice monotone properties
	that significantly simplify the proofs. We pay special attention to establishing these monotone properties.
	Our main results on MSS are established in Section~\ref{stab.sec} where we begin 
	with the simple case of SISO unstructured stochastic uncertainty. This case illustrates how 
	small-gain arguments
	similar to those used for deterministic perturbations~\cite{desoer2009feedback} 
	can be used to establish necessary and sufficient 
	MSS conditions. We then consider the structured uncertainty case, introduce 
	the loop gain operator and show that it captures the exact necessary 
	and sufficient structured
	small gain condition for MSS. Section~\ref{special.sec} examines this loop gain operator in the 
	general case as well as several special cases. We reproduce earlier results when the uncertainties 
	are uncorellated, and derive conditions for repeated uncertainties. 
	Finally, Section~\ref{perf.sec} treats the performance problem and Section~\ref{ss.sec} translates 
	our conditions to state space formulae whenever such realizations are available. 
	These can be useful for explicit computations, and in particular we provide a power iteration algorithm for 
	calculating the largest eigenvalue and corresponding eigen-matrix of the loop gain operator. 
	 We close with some remarks and 
	comments about further research questions in Section~\ref{conc.sec}.

\section{Preliminaries and Basic Results} 			\label{prelims.sec}

	All the signals  considered are  defined on the 
	half-infinite, discrete-time interval   $\Z^+ := \{0,1, \ldots \}$. The dynamical systems  
	considered are maps between various signal spaces over the time interval $\Z^+$. This is done in 
	contrast with the standard stationary stochastic processes setting over $\Z$ 
	since stability arguments involve the growth of signals
	starting from some initial time. 
	
	A stochastic process $u$ is a one-sided sequence of random variables $\{u_t; ~t\in\Z^+\}$. 
	The notation $u_t := u(t)$ is used whenever no confusion can occur due to the presence of other 
	indices. 
	Without loss of generality we assume all processes to be zero mean. 
	For any process, its 
	 {\em instantaneous variance sequence} 
	 $\ub_t := \expec{u_t^* u_t^{} }$ is denoted by small bold font, 
	 and its  {\em instantaneous covariance matrix sequence} 
	$\Ub_t =  \expec{u_t u_t^*}$ is denoted by capital bold font. 
 	The entries of $\Ub_t$ are  mutual correlations of the components of the vector $u_t$, and are sometimes 
	referred to as {\em spatial correlations}. 
	Note that $\ub_{t}  = \tr{ \Ub_t }$.

	A process $u$ is termed {\em second order} if it has finite covariances $\Ub_t$ for each 
	$t\in\Z^+$. A processes is termed {\em white} if it is uncorrelated at any two distinct times, i.e. 
	$ \expec{u_t u_\tau^*}=0$ if $t\neq \tau$. 
	Note that in the present context, a white processe $u$ may still have spatial correlations, i.e. its 
	instantaneous correlation matrix $\Ub_t$ need not be the identity.
	A process $u$ is termed {\em temporally independent} if 
	$u_t$ and $u_\tau$ are independent when $t\neq \tau$. 
	 Although the processes considered in this paper are technically not stationary (stationary processes
	are defined over the doubly infinite time axis), it can be shown that they are asymptotically stationary
	in the sense that their statistics become approximately stationary in the limit of large time, or 
	quasi-stationary in the terminology of 
	\cite{ljung1999system}. This fact is not used in the present treatment  and the preceding comment is only 
	included for clarification. 
	
\subsection*{Notation Summary} 

	\subsubsection{Variance and Covariance Sequences} 

	A stochastic process is a zero-mean,  one sided sequence of vector-valued 
	random variables $\{u_t; ~t\in\Z^+\}$.
	\begin{itemize} 
		\item The {\em variance sequence} of $u$ is
			\[  	\ub_t ~:=~ \expec{u_t^*u_t}. 	\]
		\item The {\em covariance sequence} of $u$ is
			\[  	\Ub_t ~:=~ \expec{u_t u^*_t}. 	\]		
		\item When it exists, the asymptotic {\em limit} of a covariance sequence is denoted by an overbar 
			\be	
				\Ubb ~:=~ \lim_{t\rightarrow\infty} \Ub_t,
			  \label{bar_nota.eq}
			\ee
			with similar notation for  variances  $ \ubb ~:=~ \lim_{t\rightarrow\infty} \ub_t$.
	\end{itemize}
	We use calligraphic letters $\cM$ to denote  LTI systems as operators, and capital letters $\{ M_t\}$ to denote 
	elements of  their matrix-valued impulse response sequences, i.e. $y=\cM u$ is operator notation for
	$y_t = \sum_{l=0}^t M_{t-\tau} u_\tau$. 
	\begin{itemize} 
		\item If $\cM$ has finite $H^2$ norm, then the limit 
			 of the output covariance $\Ybb$ when the input is white and has a covariance 
			sequence with a limit $\Ubb$ is denoted by 
			\[
			\hspace*{-1em}
			\begin{array}{cccl}
				\Ybb & = &   \covl{\cM}{ \Ubb}		&	\\ 
				 \rotatebox[origin=c]{-90}{$:=$} & & \rotatebox[origin=c]{-90}{$:=$}	 \rule{0em}{1.2em}	&	\\ 
				\displaystyle \lim_{t\rightarrow\infty} \Yb_t  & &  
					\displaystyle	\sum_{\tau=0}^\infty M_\tau \Ubb M^*_\tau  & =~ 
							\displaystyle	\lim_{t\rightarrow\infty} \sum_{\tau=0}^t M_\tau \Ub_{t-\tau} M^*_\tau .
			\end{array}
			\] 
			Note that $\covl{\cM}{.}$ is  a matrix-valued linear operator. 
		\item The response to spatially uncorrelated white noise is denoted by 
			\[ 	\textstyle
				  \covl{\cM }{ I }	 
						~=~ 	\sum_{\tau=0}^\infty M_\tau M^*_\tau  ~=:~ \Mbb
			\]
			
	\end{itemize}

	\subsubsection{The Hadamard Product} 
	
	For any vector $v$ (resp. square matrix $V$),  $\Diag{v}$ (resp.  $\Diag{V}$ )
	 denotes  the diagonal matrix with diagonal entries equal to those 
	 of $v$ (resp. $V$). For any square matrix $V$,  $\diag{V}$ is the {\em vector}
	 with entries equal to the diagonal entries of $V$.

	The Hadamard, or element-by-element product of two matrices $A$ and $B$ is denoted by $A\circ B$. 
	We will  use the notation 
	\[
		A^{\circ 2} ~:=~ A\circ A 
	\]
	for the element-by-element matrix square. Note that with this notation, for any matrix $V$  
	\[
			I \circ V ~=~ \Diag{V} . 
	\]
	A matrix-valued operator which  will be needed is 
	\be
		F(V) ~:=~  I \circ \big( A V A^* \big)  ~=~ \Diag{A V A^*} .
	  \label{hadop1.eq}
	\ee	
	In particular, we will need to characterize its action  
	 on diagonal matrices which is easily shown to be 
	\be
		\Diag{ A~ \Diag{v} ~A^* } ~=~  \Diag{A^{\circ 2} v }. 
	  \label{hadop2.eq}
	\ee
	In other words, if $V=\Diag{v}$ is diagonal, then the diagonal part of $AVA^*$ 
	as a vector is simply the matrix-vector product of $A^{\circ 2}$ with $v$.

\subsection{Input-Output Formulation of Mean Square Stability (MSS)} 

	Let $\cM$ be a causal linear time invariant (MIMO) system. The system $\cM$ is completely 
	characterized by its impulse response, which is a matrix valued sequence
	$\{ M_t; ~t\in \Z^+\}$. 
	The action of $\cM$ on an input signal $u$ to produce an output signal $y$ 
	is given by the convolution sum
	\be
		y_t ~=~ \sum_{\tau=0}^t M_{t-\tau} ~u_\tau , 
	  \label{Gconv.eq}
	\ee
	where without loss of generality,  zero initial conditions are assumed. 
	
	If the input $u$ is a zero-mean, second-order stochastic process, then it is clear from~\req{Gconv} that 
	$y_t$ has finite covariance for any $t$, even in the cases where this covariance may grow unboundedly 
	in time. 
	If $u$ is in addition white, then  the following calculation is standard 	
	\begin{eqnarray}
		\Yb_t & = &  
			\expec{	\lb \sum_{\tau=0}^t M_{t-\tau} u_\tau  \rb	\lb \sum_{r=0}^t u_r^* M^*_{t-r}  \rb }
																			\nonumber	\\
			& = & \sum_{\tau=0}^t \sum_{r=0}^t 	  M_{t-\tau} ~ \expec{ u_\tau u_r^*}   M^*_{t-r}	\nonumber	\\
		\Yb_t & = &   \sum_{\tau=0}^t    M_{t-\tau}  ~\Ub_\tau~ M^*_{t-\tau} 	.
														   \label{outincorr.eq}	
	\end{eqnarray}
	Note that this is a matrix convolution which relates the  instantaneous covariance sequences of the output and white input. 
	For SISO systems, this relation simplifies to 
	\be
		\yb_t ~=~ \sum_{\tau=0}^t M^2_{t-\tau}~ \ub_\tau. 
	  \label{siso_variances.eq} 
	\ee	
	For systems with a finite number of inputs and outputs, taking the trace of~\req{outincorr} gives 
	\begin{eqnarray}
		\yb_t & = & \tr{ \Yb_t   } ~=~    \sum_{\tau=0}^t    \tr{  M_{t-\tau}  ~\Ub_\tau~ M^*_{t-\tau} 	} \label{rhobyk.eq} \\
			& = &  \sum_{\tau=0}^t    \tr{  M^*_{t-\tau}   M_{t-\tau}  ~\Ub_\tau	}					  	\nonumber  \\
			& \leq &  \sum_{\tau=0}^t    \tr{  M^*_{t-\tau}   M_{t-\tau} } ~ \tr {\Ub_\tau	}						\nonumber \\
		\yb_t & \leq & \lb  \sum_{\tau=0}^\infty    \tr{  M^*_{t-\tau}   M_{t-\tau} } \rb    \lb  \sup_{0\leq t<\infty} \ub_t \rb  .
																				\label{linftybound.eq}
	\end{eqnarray}
	where the first inequality holds because for any two positive semidefinite matrices $A$ and $B$ we have ${\tr{AB} \leq \tr{A} \tr{B}}$ \cite{coope1994matrix}.
	The above calculation motivates the 
	 following input-output definition of Mean-Square Stability. 
	\begin{definition} 
		A causal linear time invariant system $\cM$ is called {\em Mean-Square Stable} (MSS) if for each 
		white, second-order  input process $u$ with uniformly bounded variance, 
		the  output process $y=\cM u$ has 
		uniformly bounded variance 
		\be
			\yb_t ~:=~ \expec{y_t^* y_t} ~\leq ~c  ~\left( \sup_\tau \ub_\tau \right),
		   \label{mssdef.eq}
		\ee
		with $c$ a constant independent of $t$ and the process $u$.
	\end{definition}	
	Note that the first term in~\req{linftybound} is just the $H^2$ norm of $\cM$
	\[
		\| \cM \|_2^2 ~:=~  \sum_{t=0}^\infty \tr{M_t M_t^*}. 
	\]
	The second term  is an $\ell^\infty$ norm 
	on the variance sequence, and thus the  bound in~\req{linftybound} can be compactly rewritten as
	\be
		\|\yb\|_\infty ~\leq~ \|\cM\|_2^2 ~\|\ub\|_\infty .
	   \label{varboundnorms.eq}
	\ee
	
	From~\req{rhobyk}, it is easy to see that equality in~\req{varboundnorms} holds when $u$ has 
	constant identity covariance ($\Ub_\tau = I$). Conversely, 
	if $\cM$ does not have finite $H^2$ norm, this input causes 
	  $\yb_t$ to
	 grow unboundedly. Thus a system is mean-square stable if and only if it has finite $H^2$ 
	 norm\footnote{It must be emphasized that this conclusion holds only if MSS is defined with 
	 as boundedness of variance sequences when the input is white.
	As is 
	 well-known from the theory of stationary stochastic processes, the instantaneous variance of a signal 
	 is the integral of its Power Spectral Density (PSD). The integrability of the output PSD cannot be concluded 
	 from the 
	 integrability of the system's magnitude squared response unless the input has a flat PSD (i.e. white). 
	 Thus for colored inputs, the boundedness of the output variance sequence cannot be concluded from only 
	  the $H^2$ norm.}. 
	 
	 For systems that have a finite $H^2$ norm, the output covariance sequence has a steady state limit when 
	 the input covariance sequence does. More precisely let $y=\cM u$, and
	 the input $u$ be such that 
	 \[
	 	\lim_{t\rightarrow\infty} \Ub_t ~=:~ \Ubb
	\]
	exists. Then if $\cM$ has finite $H^2$ norm, it follows that the output covariance has the limit 
	\be
		\Ybb ~:=~ \lim_{t\rightarrow\infty} \Yb_t ~=~ \sum_{\tau=0}^\infty M_\tau~\Ubb~M_\tau^* ~=:~ \covl{\cM }{ \Ubb } .
	  \label{covlim.eq}
	\ee
	For covariance sequences with a well-defined limit, the overbar bold capital notation is used for 
	the limit value as above. Also as above, the notation $\covl{\cM }{ \Ubb }$ is used for the 
	steady state  output covariance of an LTI system $\cM$ with input that has steady state 
	covariance of $\Ubb$. 
	When $\Ubb = I$, the following compact notation is used 
	\[
		\Mbb ~:=~ \covl{\cM }{ I } . 
	\]
	Thus $\Mbb$ is the steady state covariance of the output of an LTI system $\cM$ when the input is white 
	and has a steady state  covariance of identity.

\subsection{Mean-Square Stability of Feedback Interconnections} 
	
	The input-output setting for mean-square stability of a feedback interconnection can be 
	motivated using the conventional scheme~\cite{desoer2009feedback}
	of injecting exogenous disturbance signals into all loops.	
	\begin{definition}
	Consider the feedback system of Figure~\ref{feedback_MSS.fig} with $d$ and $w$ being
	white second order processes, and $\cM$ and $\cG$  causal LTI systems.  
	The feedback system in Figure~\ref{feedback_MSS.fig} is called mean-square stable if all signals $u$, 
	$y$, $v$ and $r$ have uniformly bounded variance sequences, i.e. if there exists a constant $c$
	such that 
	\begin{eqnarray*}
		\lefteqn{
		\max \left\{ \|\ub \|_\infty , \|\yb \|_\infty , \|\vb \|_\infty , \|\rbo \|_\infty \right\}	}
		~~~~~~~~~~~~~~~~~~~~~~& & 		\\ 
		& & ~\leq~ c ~\min\left\{ \|\db \|_\infty , \|\wb \|_\infty \right\}.
	\end{eqnarray*}
	\end{definition}
	\begin{figure}[h]
		\centering
			\includegraphics[width=0.35\textwidth]{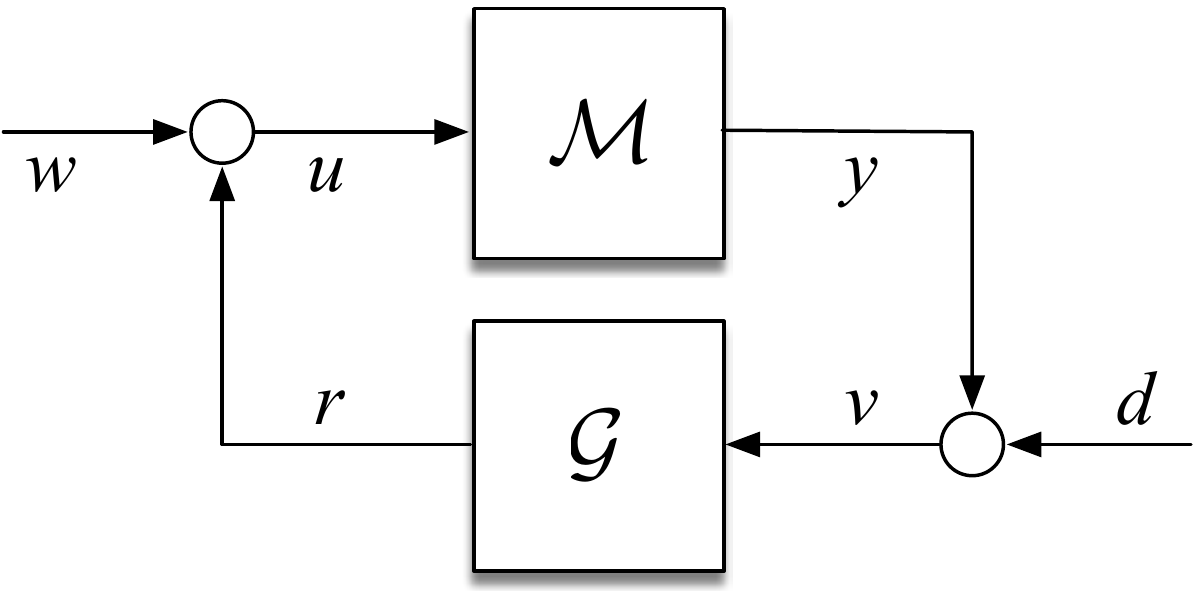}	
		\caption{The definition of mean-square stability for a feedback interconnection. The  exogenous
		disturbance signals are white random processes, and the requirement is that all signals in the loop 
		have uniformly bounded variance sequences.} 
	   \label{feedback_MSS.fig}
	\end{figure}

	\begin{remark} 										\label{colored.remark}
		A standard argument implies that the feedback interconnection is MSS iff the four mappings 
		$(I-\cM \cG)^{-1}$, $\cG(I-\cM \cG)^{-1}$, $(I-\cM\cG)^{-1}\cM$ and  $\cM\cG(I-\cM\cG)^{-1}$
		have finite $H^2$ norms. In general, it is not possible to bound those closed-loop norms in terms
		of only the $H^2$ norms of $\cM$ and $\cG$. In other words, it is not generally possible to carry 
		out a small-gain type analysis of the feedback system of Figure~\ref{feedback_MSS.fig} using only 
		$H^2$ norms. Another way to see this is that bounds like~\req{varboundnorms} are not directly 
		applicable to Figure~\ref{feedback_MSS.fig} since the signals $u$ and $v$ will not in general be 
		white.
	\end{remark} 

	Despite the above remark, in the present paper the concept of feedback stability is used when one 
	of the subsystems is a temporally independent multiplicative uncertainty. As will be seen, this 
	 has the effect of ``whitening'' (temporally de-correlating) the signal going through it,
	thus enabling a type of small-gain analysis.

\subsection{Stochastic Multiplicative Gains} 

	The MSS stability problem considered in this paper is for systems of the structure depicted 
	in Figure~\ref{M_Delta.fig},  where 
	\[
		\Gain(t) ~:=~ \Diag{  \rule{0em}{1em} \gain_1(t), ~\ldots,~ \gain_n(t) } 
	\]
	is a diagonal matrix of time-varying scalar stochastic gains acting on the vector signal $v$  
	\[
		r_t ~=~ \Gain_t ~v_t, 
	\]
	and $\cM$ is a strictly causal LTI system. 
	Without loss of generality, $\Gain$ can be assumed to be a zero mean process as the mean value can be absorbed 
	into the known part of the dynamics. The assumptions we make on the additive and multiplicative uncertain signals 
	are:
	 
	 \noindent
	 \textit{Assumptions on $d$, $w$  and $\Gain$} 
		\begin{enumerate} 
			\item The random variables $\{d_t\}$, $\{w_t\}$ and $\{\Gain_t\}$ are all mutually independent. 
			\item Both $d$ and $w$ have non-decreasing covariance sequences, i.e. 
				\[
					t_2\geq t_1 ~~\Longrightarrow~~ \Db_{t_2}\geq\Db_{t_1}, ~\Wb_{t_2}\geq\Wb_{t_1}.
				\]
		\end{enumerate} 
	
	The first assumption above
	 on the mutual independence of the $\Gain$s is crucial to the techniques used in this paper. 
	Note however that 
	for any one time $t$, individual components of $\Gain_t$ maybe correlated, and that is referred to as {\em spatial
	correlations} which can be characterized as follows. 
	 Let $\gain(t)$ 
	denote the vector 
	\[
		\gain(t) ~:=~ \bbm  \gain_1(t) & \cdots &  \gain_n(t) \ebm^* .
	\]
	The instantaneous
	correlations of the $\gain$'s can be expressed with the matrix 
	\be
		\Gainb ~:=~ \expec{ \rule{0em}{1em} \gain(t) ~\gain^*(t)}, 
	   \label{spatialcor.eq}
	\ee
	which is assumed to be independent of $t$.

	\begin{figure}[h]
		\begin{center}	\includegraphics[width=.3\textwidth]{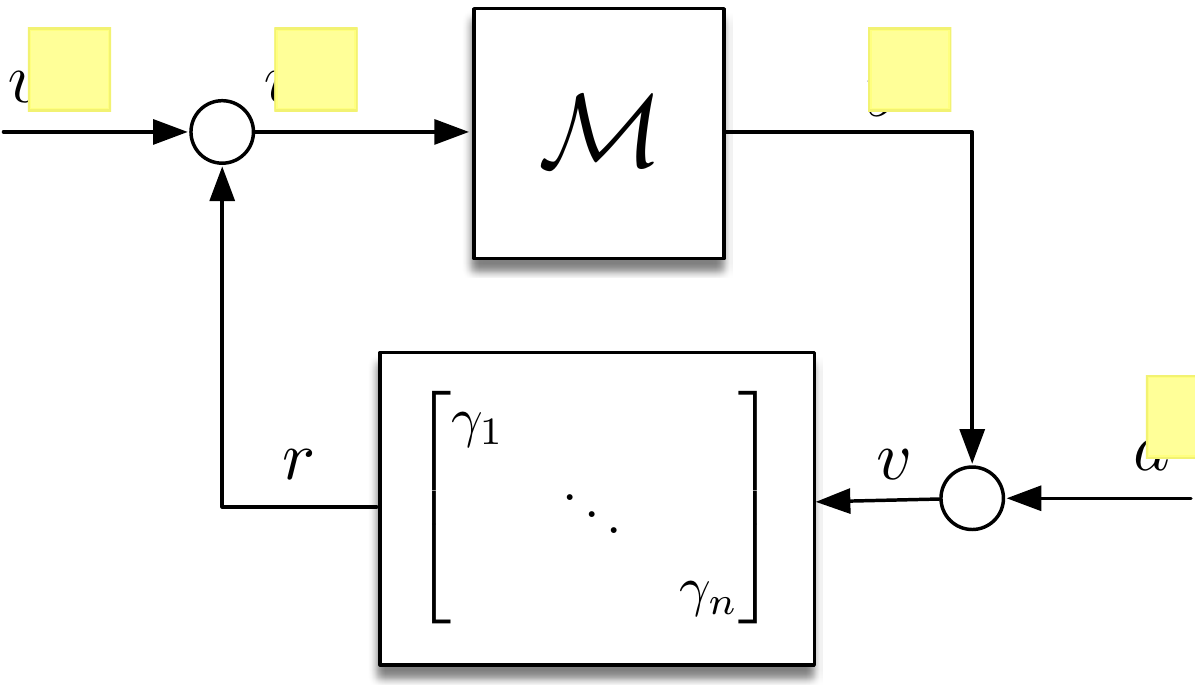}	\end{center}
		\caption{A strictly causal LTI system $\cM$ in feedback with multiplicative stochastic gains $\gain_i$'s. 
				The exogenous stochastic signals $d$ and $w$ are injected to test MSS of the feedback system, 
				which holds when all internal loop signals $u$, $y$, $v$ and $r$ have bounded variances sequences.} 
	   \label{M_Delta.fig}
	\end{figure}

	The mutual independence (in time) of the perturbations $\Gain$ and the strict causality of $\cM$ have  
	 implication for the dependencies of the various signals in the loop on the $\Gain$s. 
	 This is expressed in the following lemma whose proof is 
	found in the appendix. 
	\begin{lemma} 							\label{pastpresent.lemma}
		In the feedback diagram of Figure~\ref{M_Delta.fig}, assume $\cM$ is strictly causal, and 
		that  $\Gain_t$ and $\Gain_\tau$ are independent for $t\neq\tau$. Then we have 
		\begin{enumerate} 
			\item Past and present values of $v$ and $y$ are independent of present 
				and future values of $\Gain$, i.e. 
				\be 
					\begin{array}{lll}
						\Gain_t, ~ y_\tau,  &\tau\leq t,  &~~~\mbox{are independent}		\\
						\Gain_t, ~ v_\tau,  &\tau\leq t,  &~~~\mbox{are independent}
					\end{array}
				  \label{indep_yv.eq}	
				\ee
			\item Past values of $r$ and $u$ are independent of present and future values of $\Gain$, i.e. 
				\be 
					\begin{array}{lll}
						\Gain_t, ~ r_\tau,  &\tau< t,  &~~~\mbox{are independent}		\\
						\Gain_t, ~ u_\tau,  &\tau< t,  &~~~\mbox{are independent}
					\end{array}
				  \label{indep_ru.eq}	
				\ee
		\end{enumerate} 
	\end{lemma} 

	An important consequence of these relations  is that even if the
	input signal $v$ may in general be colored, multiplication by the $\Gain$s will cause the output 
	$r$ to be white. This can be seen from 
	\begin{eqnarray*}
		\expec{r_t r^*_\tau}  
		& = &   \expec{ \Gain_t v_t   v^*_\tau  \Gain^*_\tau} 		\\
		& = &   \expec{ \Gain_t }~ \expec{ v_t v^*_\tau    \Gain^*_\tau }  = 0,~~~~~~ \tau< t , 
	\end{eqnarray*}
	where the second equality follows from~\req{indep_yv}, i.e. the independence of $\Gain_t$ from $v_t$, $v_\tau$, 
	and $\Gain_\tau$ respectively. A similar argument shows that $r$ is uncorrelated with present and past values of 
	$v$ and $y$, and uncorrelated with past values of $u$, but we will not need these facts in the sequel.

	To calculate the instantaneous spatial correlations of $r$
	\begin{eqnarray}
		\expec{r_t r_t^*} 
		& = &  \expec{ \Gain_t v_t   v^*_t  \Gain^*_t} 				\nonumber		\\ 
		 & = &  \expec{ \Gain_t ~\left( \expec{v_t   v^*_t} \right)  ~\Gain^*_t} ,	\label{DMD.eq}
	\end{eqnarray}
	where the last equality follows from the independence of $v_t$ and $\Gain_t$ and 
	formula~\req{mat_indep} in Appendix~\ref{indep.appen}. 
	It is thus required to calculate quantities like
	$\expec{\Gain M \Gain}$  for some  constant matrix $M$. The case of diagonal $\Gain$ reduces to 
	\begin{eqnarray*}
		\expec{\Gain M \Gain^*} 
		& = & 		
		\expec{ \thbthtight{\gain_1}{}{}{}{\ddots}{}{}{}{\gain_n} 
		M   \thbth{\gain_1}{}{}{}{\ddots}{}{}{}{\gain_n}   }				\\
		& = & 
		\left[ \rule{0em}{1em}  m_{ij}~ \expec{\gain_i \gain_j} \right] ~=~ 
		\Gainb \circ M ,		
	\end{eqnarray*}
	which is the  Hadamard (element-by-element) product of $\Gainb $  and $ M$.
	Applying this to~\req{DMD}, the above arguments lead to the following conclusion. 
	\begin{lemma} 							\label{delta.lemma}
	Consider  the feedback system of Figure~\ref{M_Delta.fig} with $\cM$ a strictly causal LTI system, and $\Gain$ 
	  diagonal stochastic perturbations 
	  with spatial correlations~\req{spatialcor}.
	If the perturbations $\Gain$ are temporally independent, 
	then the output $r$ is a white process with 
	 instantaneous 
	spatial correlations  given by 
	\be
		\Rb_t
		~=~  \Gainb \circ \Vb_t , 
	   \label{hadamard.eq}	
	\ee 
	 the  Hadamard (element-by-element) product of $\Gainb $  and $\Vb_t$.
	\end{lemma}

	Two special cases are worth noting. If $\Gain=\gain$ is a scalar perturbation, then~\req{hadamard} 
	reduces to 
	\be
		\rbo_t ~=~ \Gainb ~ \vb_t. 
	  \label{siso_delta_variances.eq}
	\ee
	Thus multiplication by a scalar perturbation simply scales the variance of the input signal and ``whitens'' it. 
	In the special case where the perturbations are uncorrelated and all have unit variance, 
	i.e. $\Gainb ~=~ I$, a simple expression results 
	\[
		\Rb_t 
		~=~ \diag{ \Vb_t}, 
	\] 	
	where $ \diag{ \Vb_t}$ is a diagonal matrix made up of the diagonal entries of the matrix 
	$\Vb_t$. Thus 
	if the $\gain$'s are white and mutually uncorrelated, then the  vector output signal $r$ is temporally 
	and spatially uncorrelated even though $v$ may have both types of correlations. In other words, 
	a structured perturbation with uncorrelated components will spatially and temporally whiten its input.

\subsection{Covariance Feedback System} 							\label{cov.subsec}

	An important tool used in this paper is to replace the analysis of the original
	stochastic system of Figure~\ref{M_Delta.fig}
	with an equivalent 
	system that operates on the respective signals'  instantaneous covariance matrices. 
	This deterministic system 
	is depicted in Figure~\ref{M_Delta_cov.fig}. 
	\begin{figure}[h]
		\begin{center}	\includegraphics[width=.47\textwidth]{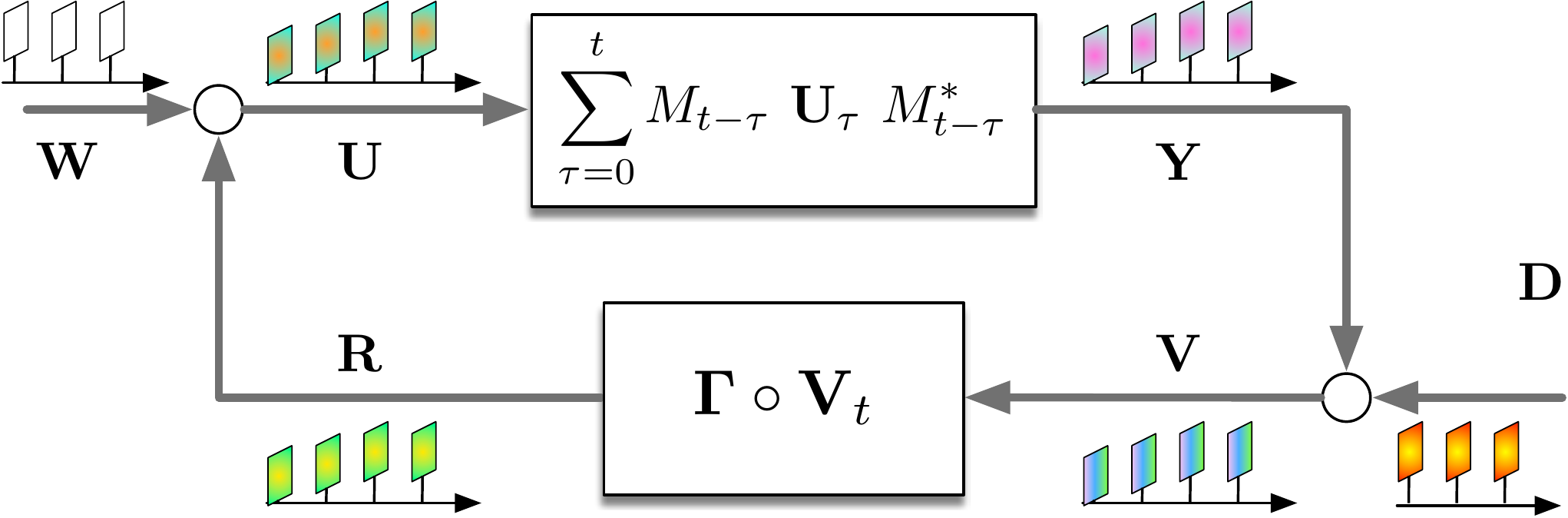}	\end{center}
		\caption{A deterministic feedback system detailing how the dynamics of the stochastic feedback 
		system in Figure~\ref{M_Delta.fig} operates on covariance sequences. Each signal in this diagram is 
		positive semi-definite matrix valued. 
		The forward block is an LTI matrix convolution, and the feedback block is the Hadamard matrix product. 
		It is notable that 
		all input-output mappings of this feedback system are monotone even if the 
		original system is not. This implies in particular that all covariance sequences in the loop are
		non-decreasing if the exogenous inputs are non-decreasing (in the semi-definite ordering on matrices).} 
	   \label{M_Delta_cov.fig}
	\end{figure}
	Each signal in this feedback system is {\em matrix-valued}. The mathematical operations on each 
	of the signals depicted in the individual blocks follow from~\req{outincorr} and ~\req{hadamard} and 
	the following observations
	\begin{enumerate} 
		\item $u$ is a white process. 
		\item For each $t$ 
			\be
				\Ub_t ~=~  \Rb_t ~+~ \Wb_t. 			\label{Ut.eq}
			\ee
		\item For each $t$ 
			\be
				\Vb_t ~=~ \Yb_t ~+~ \Db_t.			\label{Vt.eq}
			\ee
	\end{enumerate} 
	Observation 1 and 2 follow from 
	\begin{multline*}
		\expec{u_t u^*_\tau  }
		~ = ~   \expec{ (r_t+w_t) (r^*_\tau+w^*_\tau)  } 			\\
		~=~ \expec{ r_t r^*_\tau}   +\expec{r_t w^*_\tau} +  \expec{ w_t r^*_\tau}  + \expec{w_t w^*_\tau} \\ 
		~ = ~    
				\left\{ \begin{array}{ll} 
		 			 0 + 0 +  0 + 0  = 0 ,  &  \tau<t \\ 
					 \expec{ r_t r_t} +0+0+   \expec{w_t w^*_t}  =   	 \Rb_t +  \Wb_t,	
					  										& \tau=t 
					  \end{array}	\right.		
	\end{multline*}
	where $\expec{ w_t r^*_\tau}=0$ since $w$ is uncorrelated with past system
	signals, and $\expec{r_t w^*_\tau}=0$ follows from Lemma~\ref{pastpresent.lemma} because $\expec{r_t w^*_\tau} = \expec{\Gain_t v_t w^*_\tau} = \expec{\Gain_t} \expec{v_t w^*_\tau} = 0$. For the case
	$\tau<t$,     $\expec{w_t w^*_\tau}=0$ and $\expec{ r_t r^*_\tau}=0$ since  $w$ (by assumption)  and 
	$r$ (Lemma~\ref{delta.lemma})  are white respectively.  Observation 3 follows immediately from 
	\[
		\Vb_t ~=~ \expec{ (y_t+d_t)(y_t^*+d_t^*)} ~=~ \Yb_t + \Db_t, 
	\]	
	since $y_t$ is a function of $w$, $\Gain$ and past $d$'s, and thus is independent of $d_t$. 
	Note that $v$ may in general be colored. 
	
	Observation 1 implies that the forward block can indeed be written using~\req{outincorr} (the input needs to be white
	for its validity). Observations 2 and 3 imply that the summing junctions in Figure~\ref{M_Delta.fig} 
	can indeed be replaced by  summing junctions on the corresponding variance sequences in 
	Figure~\ref{M_Delta_cov.fig}.
	
\subsection{Monotonicity}						\label{monotone.subsec}

	Although it is not standard to consider systems operating on matrix-valued signals, it is rather advantageous
	in the current setting. 
	In this paper, the order relation used on matrices is always the positive semidefinite ordering (i.e. $A\geq B$ means 
	$(A-B)$ is positive semi-definite). The statements in this section apply to orderings with other positive cones, though 
	this generality is not needed here. 
	
	\subsubsection*{Monotone Operators}
	
	\begin{definition} 
		A matrix-valued linear operator $\cL:\R^{n\times n} \rightarrow \R^{m\times m}$ is called {\em monotone} 
		(in the terminology of~\cite{tigelaar1991monotone},  
		or {\em cone invariant} in the terminology of~\cite{berman1994nonnegative,parrilo2000cone}) if 
		\[
			X\geq 0 ~~\Rightarrow~~ \cL(X) \geq 0. 
		\] 
	\end{definition}
	In other words, if it preserves the semi-definite ordering on matrices (this definition is equivalent to the 
	statement $X\leq Y ~\Rightarrow  \cL(X) \leq \cL(Y)$). There is a Perron-Frobenius theory for such 
	 operators which gives them  nice  properties, 
	some of which are now summarized. 
	\begin{theorem}													\label{Ialpha.thm}
	For a matrix-valued monotone operator $\cL$ 
		\begin{enumerate} 
			\item $\exists$ a real, largest eigenvalue:  $\rho\lb\cL\rb$ is an eigenvalue of $\cL$. 
			\item $\exists$ an eigen-matrix $X\geq 0$ for the largest eigenvalue, i.e.
				\be	\cL(X) ~=~ \rho\lb\cL\rb ~X. 							\label{PF_matrix.eq}	\ee
			\item Sums and compositions of monotone operators are  monotone.  
				 For $\rho\lb\cL\rb<\alpha$,  the operator $\lb I - \cL /\alpha\rb^{-1}$ exists and is  
				monotone.											
		\end{enumerate}
	\end{theorem} 
	\begin{proof}
	The first two statements are from~\cite[Thm 2]{parrilo2000cone} or \cite[Thm 3.2]{berman1994nonnegative}. 
	That sums and compositions of monotone operators are monotone is immediate from the definition. 
	Furthermore, note that the Neuman series 
	\[
		\lb I -  \cL/\alpha \rb^{-1} ~=~ \sum_{k=0}^\infty \lb \cL/\alpha\rb^k 
	\]
	is made up of sums of compositions of a monotone operator $\cL/\alpha$. This series 
	converges in any operator norm since $\rho\lb\cL/\alpha\rb<1$ (this follows from Gelfand's formula which 
	implies that for any operator norm, there is some $k$ such that $\|\lb\cL/\alpha\rb^k \|<1$). 
	\end{proof} 
	
	Note that 
	the ``eigen-matrix'' $X$ in~\req{PF_matrix} is the counterpart of the Perron-Frobenius eigenvector for matrices 
	with non-negative entries. Such eigenmatrices will play an important role in the sequel as a sort of worst-case
	covariance matrices. 
	
	\subsubsection*{Monotone Systems and Signals}

	The positive semi-definite ordering on matrices induces a natural
	ordering on  {\em matrix-valued signals}, as well as a notion of monotonicity on systems~\cite{angeli2003monotone}. 
	For two matrix-valued signals $\Ub$ and $\Wb$ the following point-wise order relation can be defined
	\be
		(\Ub \leq \Wb)
		~~~~~\Longleftrightarrow~~~~~
		\forall t\in\Z^+, ~~ \Ub(t) \leq \Wb(t). 
	   \label{pointwise_ordering.eq}
	\ee
	For systems, the following is a restatement of the definition from~\cite{angeli2003monotone} when the initial conditions
	are zero. 
	\begin{definition} 
	An input-ouput system $\cM$ mapping on matrix-valued signals 
	is said to be {\em monotone} if whenever 
	\be
		\begin{array}{rcl} 	\Yb  &=&  \cM (\Ub) \\ \Zb  &=& \cM (\Wb) \end{array}, 
		~~~~~~~ \mbox{then}~~
		\Ub \leq \Wb ~\Rightarrow~ \Yb\leq \Zb. 
	  \label{mondef.eq}
	\ee
	\end{definition} 
	In other words, if $\cM$ preserves the positive semi-definite ordering on matrix-valued signals. 
	There is a further notion of monotonicity of an individual signal in the sense of mapping the time-axis ordering 
	to that of the matrix ordering. 
	\begin{definition} 
		A matrix-valued signal $\{\Ub(t)\}$ is said to be {\em monotone} (or {\em non-decreasing}) if 
		\[
			t_1\leq t_2 ~~\Rightarrow~~ \Ub(t_1) \leq \Ub(t_2) . 
		\]
	\end{definition}
	It is simple to show (Appendix~\ref{tempord.appen}) that a time-invariant monotone system maps non-decreasing 
	signals to non-decreasing signals. 
	
	\subsubsection*{Monotonicity of Covariance Feedback Systems}

	That the forward loop in Figure~\ref{M_Delta_cov.fig} is monotone is immediate since
	\begin{multline*}
		\forall t\in\Z^+, ~~\Ub_t \leq \Wb_t
		~~~~\Longrightarrow~~~~ 
					\forall t\in\Z^+, ~~\\
		\lb	\sum_{\tau=0}^t    M_{t-\tau}  ~\Ub_\tau~ M^*_{t-\tau} 	\rb
		\leq
		\lb	\sum_{\tau=0}^t    M_{t-\tau}  ~\Wb_\tau~ M^*_{t-\tau} 	\rb .
	\end{multline*} 
	Note that this is always
	the case {\em even when the original LTI system $\cM$ is not monotone!} It is also true that the hadamard product is
	monotone. This follows from the
	Schur Product Theorem~\cite[Thm 2.1]{horn1992block} which states that for any matrices 
	\be
		A_1\leq A_2 ~\mbox{and}~ B\geq 0  ~~\Rightarrow ~~ B \circ A_1 ~\leq~ B \circ A_2. 
	   \label{Schur.eq}
	\ee	
	Thus each of the two systems in the feedback loop 
	of Figure~\ref{M_Delta_cov.fig} is monotone. 
	
	When monotone systems are connected together with positive feedback, then all input-output mappings 
	in the resulting feedback system are also monotone 
	(see Theorem~\ref{monfdbk.thm} in Appendix~\ref{monotone.appen}). It then follows 
	(Appendix~\ref{tempord.appen}) that if the covariance signals $\Db$ and $\Wb$ are are non-decreasing,
	then all other covariance signals in feedback system are
	also non-decreasing. These non-decreasing covariance sequences are depicted in Figure~\ref{M_Delta_cov.fig}.


\section{Mean-Square Stability Conditions} 					\label{stab.sec}

	This section contains the main result of this paper characterizing MSS in terms of the spectral radius of a 
	matrix-valued operator.  MSS stability 
	arguments in the literature are typically done with state space models, and the present paper develops 
	an alternative input-output approach more akin to small-gain type analysis. This technique is most easily 
	demonstrated with the SISO case, which for clarity of exposition is treated separately first. 
	 The MIMO structured case is then developed with a similar small-gain analysis. 
	However, additional issues of spatial correlations appear in the structured case, and those are treated in 
	detail. The section ends by demonstrating how known conditions for  uncorrelated uncertainties  can 
	be derived as a special case of the general result presented here, as well as some comments about which 
	system metrics characterize MSS in the general correlated uncertainties case. 

\subsection{SISO Unstructured Uncertainty}

	Consider the simplest case of uncertainty analysis depicted in Figure~\ref{unstruct.fig} ({\em Left}). 
	\begin{figure}[h]
		\begin{center}	
		\includegraphics[height=0.075\textheight]{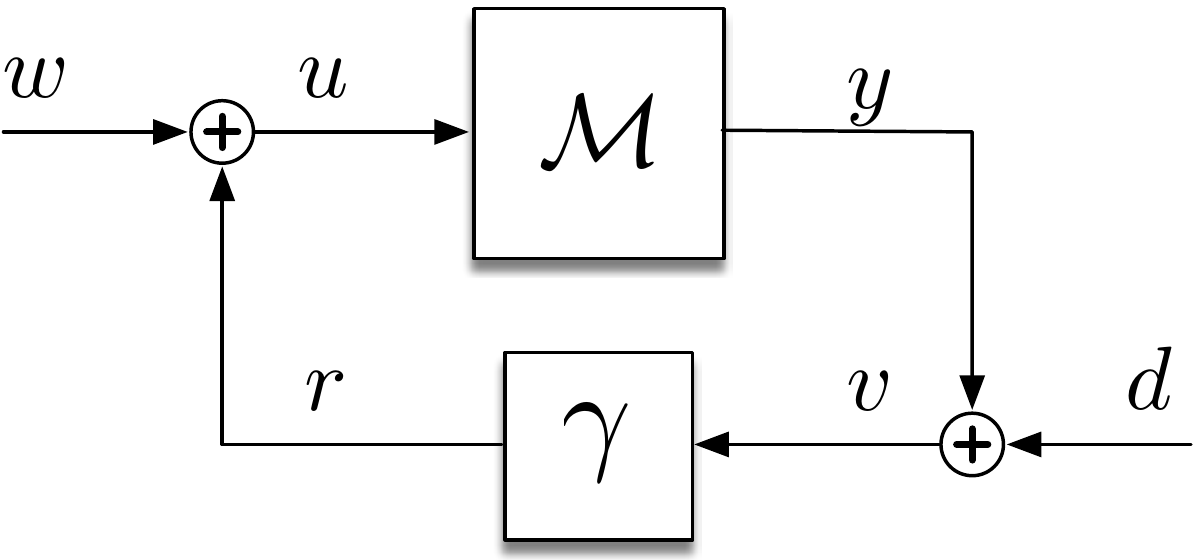}	~
		\includegraphics[height=0.075\textheight]{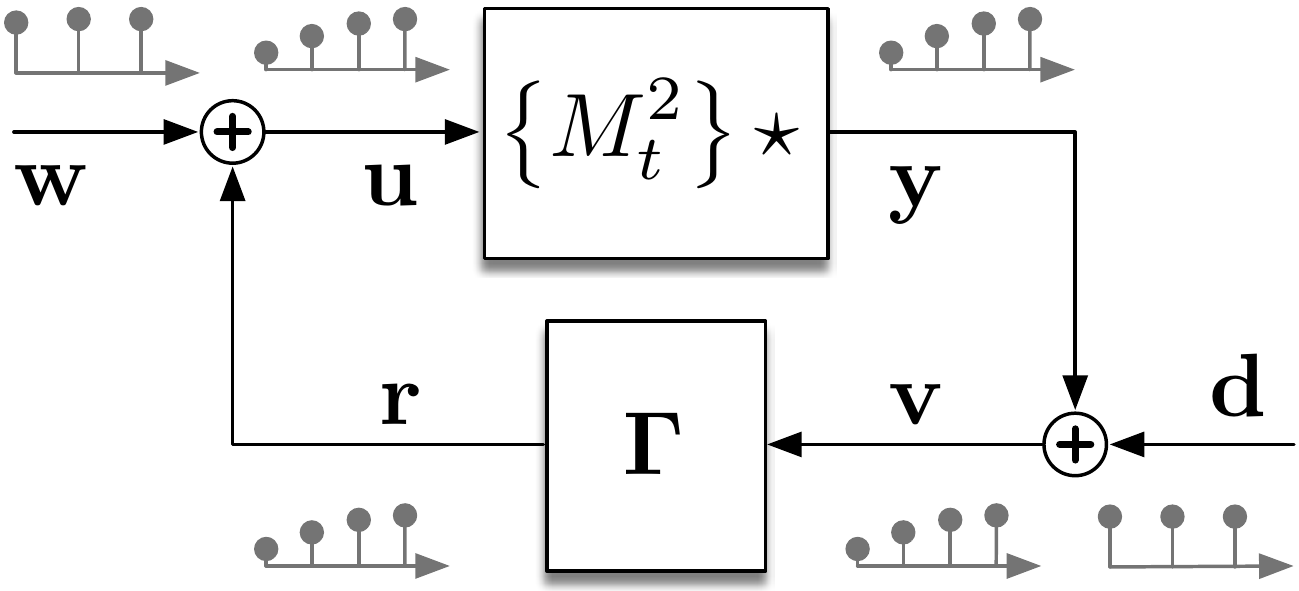}
		\end{center}
		\caption{({\em Left}) A LTI system $\cal M$ in feedback with a time-varying stochastic uncertainty
		$\left\{\gamma_t\right\}$, and with additive exogenous 
		stochastic disturbances $w$ and $d$. ({\em Right}) The equivalent LTI systems~\req{equiv_LTI} 
		operating on the variance sequences 
		of all respective signals. $\cM$ is replaced by convolution with the sequence $\left\{ M_t \right\}$, while 
		the stochastic gain $\gamma$ is replaced by multiplication with $\mathbf{\Gamma}$, its variance.  } 
	   \label{unstruct.fig}
	\end{figure}
	$\cM$ is a strictly causal LTI system, 
	$d$ and $w$ are exogenous white processes with uniform variance, 
	and $\gain$ is a white process with uniform variance $\Gainb$ and 
	independent of the signals $d$ and $w$. $\cM$ is assumed to have finite $H^2$ norm. 
		
	A small-gain stability analysis in the spirit of~\cite{desoer2009feedback} can be accomplished by 
	deriving the equivalent relations between the variance sequences of the various signals in the loop
	(Figure~\ref{unstruct.fig},  ({\em Right})). To begin with, recall the observations made in Section~\ref{cov.subsec}
	that $u$ is white, and 
	 for the SISO case the variance sequences satisfy
	\begin{eqnarray}
		\ub_t  & = &   \rbo_t  ~+~ \wb_t .	\label{rho_uk.eq}		\\
		\vb_t & = &  \yb_t ~+~ \db_t .  	 \label{sume.eq}
	\end{eqnarray}
%
	Since $u$ is white, the formulas for the variances sequences are particularly simple
	according to~\req{siso_variances} and~\req{siso_delta_variances}, 
	the equivalent relations are 
	\be
		\yb_t ~=~ \sum_{\tau=0}^t M^2_{t-\tau}~ \ub_\tau, ~~~~~~ 
		\rbo_t ~=~ \Gainb ~ \vb_t. 
	  \label{equiv_LTI.eq}	
	\ee
	The main stability result for unstructured stochastic 
	perturbations can now be stated. 
	\begin{lemma} 		\label{siso.lemma}
	Consider the system in Figure~\ref{unstruct.fig} with $\cM$ a strictly causal, stable LTI system
	and $\gain$ a temporally independent process with variance $\Gainb$. The feedback system is 
	Mean-Square Stable (MSS) if and only if 
	\[
		\|\cM\|_2^2  ~< ~ 1 / \Gainb . 
	\] 
	\end{lemma}
	\begin{proof} 		
	``\textit{if}'') This is similar to standard sufficiency small gain arguments, but using 
	variances rather than signal norms. Starting from $\ub$,  
	going backwards through the loop yields
	\begin{eqnarray} 
		\|\ub\|_\infty & \leq &  \| \rbo \|_\infty ~+~ \|\wb\|_\infty  		\nonumber		\\
				& \leq &  \Gainb~ \| \vb \|_\infty ~+~ \|\wb\|_\infty 		\nonumber  \\
				& \leq &  \Gainb~ \lb  \| \yb \|_\infty  +   \|\db\|_\infty \rb ~+~ \|\wb\|_\infty \nonumber	\\
				& \leq &  \Gainb~ \|\cM\|_2^2 ~ \| \ub \|_\infty  + \Gainb  \|\db\|_\infty  ~+~ \|\wb\|_\infty 	,
							~~~~~~
	\end{eqnarray}
	where subsequent steps follow from the triangle inequality and using~\req{varboundnorms} and
	\req{siso_delta_variances}. 
	This bound together with the assumption  $\Gainb ~ \|\cM\|_2^2  < 1$  
	gives a bound for the internal signals $u$  in terms of the exogenous signals
	$d$ and $w$ 
	\[
		\|\ub\|_\infty ~ \leq ~  \frac{1}{1- \Gainb \|\cM\|_2^2}  \left(
						\Gainb  \|\db\|_\infty +  \|\wb\|_\infty \right) ,
	\]
	In addition, this bound gives bounds on variances of the remaining internal signals $y$, $v$ and $r$ 
	as follows from~\req{varboundnorms}, \req{sume}, and \req{siso_delta_variances} respectively.

	\textit{``only if''})  See Appendix~\ref{nec1.appen}. 
	\end{proof}
	
	Two remarks are in order regarding the necessity part of the previous proof. 
	First is that there was no need to construct a so-called ``destabilizing'' perturbation as is typical in 
	 worst case perturbation analysis.  Perturbations here  are described statistically rather than 
	 members of sets, and variances 
	will always grow when the stability condition is violated. 
	Second, the necessity argument can be interpreted as showing that $\|\cM\|_2 \geq 1$ 
	implies that the transfer function $(1-\cM(z))$ has a zero in the interval $[0,\infty)$, and thus 
	$(1-\cM(z))^{-1}$ has an unstable pole. The argument presented above however is more easily 
	generalizable to the structured MIMO case  considered  next.

\subsection{Structured Uncertainty}

	In the analysis of  MIMO structured uncertainty, a certain matrix-valued operator will play 
	a central role, and therefore it is first introduced and some of its properties investigated. The main 
	result on MSS for the structured case is then stated and proved. 

	\subsubsection*{The Loop Gain Operator} 
		Consider the matrix-valued operator
		\be
			\cL (X) ~:=~  \Gainb \circ  \left( 	\sum_{t=0}^\infty M_{t} ~X~ M^*_{t} \right),
		  \label{LGO.eq}	
		\ee
		where $\Gainb$ is the correlation matrix of the uncertainties~\req{spatialcor},  and $\{M_t\}$ is the 
		matrix-valued impulse response sequence of a stable (finite $H^2$ norm), causal,  LTI system $\cM$. 
		This is termed the {\em loop gain operator} since it 
		is how the covariance matrix of stationary white noise input $u$ is mapped to the covariance of the signal 
		$r$ (which will also be white) in
		Figure~\ref{M_Delta_cov.fig}, i.e. it describes what happens to the instantaneous covariance matrix of white noise as 
		it goes once around the loop. 

		The loop gain operator is monotone since it is the composition of two monotone operators. There will also be 
		need to consider finite-horizon truncations of it 
            	\be
            		\cL_T \left(X \right) ~:=~ \Gainb \circ  \left( 
            				\sum_{t=0}^T M_t X ~M^*_t   \right) .
            	   \label{cldef.eq}
            	\ee
		If $\cM$ has finite $H^2$ norm, then it is immediate that 
		\[
			\lim_{T\rightarrow\infty} \cL_T ~=~ \cL
		\] 
		in any (finite-dimensional) operator norm. 

		It turns out that the spectral radius of the loop gain operator is the exact condition for  MSS. This is stated next.

	\begin{theorem} 						\label{main_MIMO_mss.thm}
		Consider the system in Figure~\ref{M_Delta.fig} where $w$ is a white process, both $d$ and $w$ 
		have bounded,  monotone
		covariance sequences, $\Gain$ is  temporally independent multiplicative noise with instantaneous correlations
		$\Gainb$, and $\cM$ is a 
		 stable (finite $H^2$ norm), strictly-causal,  LTI system. 
		The feedback system is Mean 
		Square Stable (MSS) if and only if 
		\[
			\rho \left( \cL \right) ~<~ 1, 
		\]
		where $\cL$ is the matrix-valued ``loop gain operator'' defined in ~\req{LGO}.
		
	\end{theorem} 	
	\begin{proof}	``\textit{if}'') 
		Recalling the observations made in Section~\ref{cov.subsec}, 		
		an expression for 	
		 $\Ub_t$ can be derived by following signals backwards through the loop in Figure~\ref{M_Delta_cov.fig}
            	\begin{eqnarray} 
            		\Ub_t & = &  \Rb_t + \Wb_t ~=~  \Gainb \circ \Vb_t + \Wb_t
            														\nonumber		 \\
            			 & = &    \Gainb \circ  \left(  \Yb_t + \Db_t \right) + \Wb_t
            														\nonumber		 \\
            		\Ub_t & = & \Gainb \circ  \left( \sum_{\tau=0}^t M_\tau \Ub_{t-\tau} M^*_{\tau}  + \Db_t  \right)
									+ \Wb_t ,
            				\label{sigmau.eq}
            	\end{eqnarray} 
		which follow from~\req{Ut}, \req{hadamard}, \req{Vt}, the fact that $u$ is white 
		and~\req{outincorr} respectively. 
		The monotonicity (Section~\ref{monotone.subsec}) of the feedback system in Figure~\ref{M_Delta_cov.fig} relating 
		covariance sequences implies that 
            	$\Ub$ is a non-decreasing sequence. This gives the following bound 
            	\[
            		\sum_{\tau=0}^t M_{\tau} \Ub_{t-\tau} M^*_{\tau} 
            		~\leq~ 
            		\sum_{\tau=0}^t M_{\tau} \Ub_t M^*_\tau 	, 	
            	\]
		which together with Schur's theorem~\req{Schur} allows for
            	replaceing~\req{sigmau} with the bounds
            	\begin{eqnarray}
            		\Ub_t 
            		 & \leq & 
            		\Gainb \circ  \left( 	\sum_{\tau=0}^t M_{\tau} \Ub_t M^*_\tau   \right)
            		 + \Gainb \circ \Db_t 	+ \Wb_t									\nonumber	\\
            		 & \leq & 
            		\Gainb \circ  \left( 	\sum_{\tau=0}^\infty M_\tau \Ub_t M^*_\tau   \right)
            		+ \Gainb \circ \Db_t 	+ \Wb_t 			.						\nonumber
            	\end{eqnarray}
            	To see how this inequality gives a uniform bound on the sequence $\Ub$, rewrite it using the definition of $\cL$ as 
            	\[
            		\Big( I ~-~ \cL \Big)  (\Ub_t) ~\leq~   \Gainb \circ \Db_t + \Wb_t .
            	\]
		Now  $\rho\lb\cL\rb<1$ implies (by Theorem~\ref{Ialpha.thm} (3)) that $\lb I - \cL \rb^{-1}$ exists and is 
		a monotone operator, and therefore 
            	\begin{eqnarray*}
            		 \Ub_t   
            		& \leq &   \lb I - \cL \rb^{-1}   \Big(  \Gainb \circ \Db_t +\Wb_t  \Big)		\\
            		& \leq &   \lb I - \cL \rb^{-1}   \Big(  \Gainb \circ \bar{\Db} + \bar{\Wb}  \Big),
            	\end{eqnarray*}
		where the first inequality follows from the monotonicity of $\lb I - \cL \rb^{-1}$, and the second inequality follows 
		from its linearity, Schur's theorem, and replacing $\Db_t$ and $\Wb_t$ by their steady state limits. This provides
		a uniform upper bound on the sequence $\Ub$, and 
		note that
            	the stability of $\cM$ then implies in addition that all other signals in Figure~\ref{M_Delta_cov.fig} are uniformly 
		bounded.

		\textit{``only if''}) 
		In a similar manner to the necessity proof of Lemma~\ref{siso.lemma}, it is shown that $\rho(\cL)\geq 1$ implies that
		$\Ub$ has an unbounded subsequence. First, observe that by setting $\Db_t=0$, ~\req{sigmau} gives the following bounds  
            	\begin{eqnarray} 
            		\Ub_{T k} & = & \Gainb \circ  \left( 
            					\sum_{\tau=0}^{Tk} M_{Tk-\tau} \Ub_{\tau} M^*_{Tk-\tau}   \right)  + \Wb_{Tk} 	\nonumber	\\
            				& \geq & \Gainb \circ  \left( 
            					\sum_{\tau=T (k-1)}^{T k} M_{Tk-\tau} \Ub_\tau M^*_{Tk-\tau}   \right)  + \Wb_{Tk}  	\nonumber	\\
            			& \geq & \cL_{T} \left( \Ub_{T (k-1)} \right) + \Wb_{Tk} ~.							\label{LTlb.eq}
            	\end{eqnarray} 
		where the first inequality follows from Schur's theorem~\req{Schur}, and the second inequality follows from 
		the monotonicity of the sequence $\Ub$ and the monotonicity of the operator $\cL_T$.

		A simple induction argument exploiting the monotonicity of $\cL_T$ yields
		\begin{align} \label{recursion0.eqn}
		\Ub_{Tk} \geq \cL_T^k\left(\Ub_0\right) + \sum_{r=0}^{k-1} \cL_T^r \left(\Wb_{T(k-r)} \right).
		\end{align}	
		Now, set the exogenous covariance $\Wb_{Tk} = \Ubwh$, where $\Ubwh$  (the Perron-Frobenius eigen-matrix) is the non-zero semi-definite eigen-matrix such that  $\cL (\Ubwh) =  \rho(\cL) ~\Ubwh$ (Theorem~\ref{Ialpha.thm} (2)). Note that the initial covariance is thus $\Ub_0= \Db_0 = \Ubwh$. 
		Substituting in (\ref{recursion0.eqn}) yields
		\begin{align} \label{recursion0IC.eqn}
		\Ub_{Tk} \geq \sum_{r=0}^k \cL_T^r \left(\Ubwh \right).
		\end{align}
		Since $\lim_{T\rightarrow\infty} \cL_T(\Ubwh) =\cL(\Ubwh) = \rho(\cL) \Ubwh$,  then for any ${\epsilon>0}$, 
		$\exists ~T>0$ such that $||\rho(\cL) \Ubwh - \cL_T(\Ubwh)|| \leq \epsilon ||\Ubwh||$. This inequality coupled with the fact that $0 \leq \cL_T(\Ubwh) \leq \rho(\cL) \Ubwh$ allows us to apply Lemma~\ref{limit_ordering.lemma} to obtain
		\be
		\cL_T(\Ubwh)     ~\geq~     \lb \rho\lb\cL\rb - \epsilon c\rb   ~ \Ubwh ~=:~    \alpha~ \Ubwh, 
		\label{cLTbound.eq}
		\ee
		where $c$ is a positive constant that only depends on $\Ubwh$ (Lemma~\ref{limit_ordering.lemma}). 
		Then, by (\ref{recursion0IC.eqn}), the one-step lower bound~\req{cLTbound} becomes
		\be
		\Ub_{T k} ~\geq~ \lb  \sum_{r=0}^k \alpha^r \rb ~ \Ubwh  
		~=~ \frac{\alpha^{k+1}-1}{\alpha-1} ~ \Ubwh.
		\label{Usub_growth.eq}
		\ee
		First consider the case when $\rho\lb\cL\rb>1$, then $\epsilon$ can be chosen small enough so that $\alpha>1$ and therefore $\Ub$ is a geometrically growing 
		sequence.

		The case $\rho(\cL)=1$ can be treated in exactly the same manner as in the proof of Lemma~\ref{siso.lemma}
		to conclude that $\Ub$ has a (not necessarily geometrically) unboundedly growing subsequence. 
	\end{proof}

		{\em The worst case covariance matrix}: 
		An interesting contrast between the SISO and MIMO case appears in the necessity argument above. 
		Comparing the expressions for the unbounded sequences~\req{alphabound} and~\req{Usub_growth}, 
		it appears that the Perron-Frobenius eigen-matrix $\Ubwh$ of $\cL$ represents a sort of worst case 
		growth covariance matrix. In other words, to achieve the highest rate of covariance growth in the 
		feedback system, one needs the input $w$ to have spatial correlations such that $\Wb_t = \Ubwh$. 
		In an analysis where there are no exogenous inputs and one only considers growth of initial state 
		covariances, there is a similar worst-case initial state covariance that corresponds to $\Ubwh$. 
		Section~\ref{ss.sec} elaborates on this point.

\section{Special Structures and Representations} 						\label{special.sec}
	
	We consider first the special case of uncorrelated uncertainties, and show how the well-known result follows as a 
	special case. We then look at a Kronecker product representation of the general case which clarifies the 
	role played by system metrics other than $H^2$ norms in MSS conditions. These metrics involve what might be 
	termed as {\em autocorrelations between subsystems impulse responses}. 
	 Finally we consider the case of circulant 
	systems in which the presence of spatial symmetries provides conditions of intermediate difficulty between the 
	uncorrelated and the general case. 

\subsection{Uncorrelated Uncertainties}		 \label{diagappen}

	A well-known result in the literature~~\cite{lu2002mean,elia2005remote,hinrichsen1995stability,el1992stability}
	is the case of uncorrelated uncertainties $\{\gain_i\}$, where the MSS condition is known to be given 
	by the 
	spectral radius of the matrix of $H^2$ norms of subsystems of $\cM$. 
	We now demonstrate how this
	 result follows directly 
	as a special case of Theorem~\ref{main_MIMO_mss.thm}. 
	 
	For uncorrelated uncertainties, $\Gainb = I$, and the loop gain operator~\req{LGO} becomes 
	\be	
		\cL(X) ~:=~ \Diag{ \sum_{t=0}^\infty M_t ~X~ M^*_t  }. 
	   \label{cLdiag.eq}
	\ee
	where $\Diag{Y}$ is a diagonal matrix made up of the diagonal part of $Y$. In this case, 
	any eigen-matrix of $\cL$ (corresponding to a non-zero eigenvalue) 
	must clearly be a diagonal matrix, so it suffices to consider how $\cL$ acts on diagonal matrices. 
	Let $V := \Diag{ v }$  be a diagonal 
	matrix,  and recall the characterization~\req{hadop1}-\req{hadop2} of terms like  
	$\Diag{H V H^*} $ on diagonal matrices. Applying this term-by-term to the sum in~\req{cLdiag} gives 
	\begin{align}
		V=\Diag{v}, ~~\Rightarrow~~ 
		\cL(V) &=~ \Diag{ \sum_{t=0}^\infty M^{\circ 2}_t ~v  } 			\nonumber		\\
				&  =: ~ \Diag{ \Mb^{\circ} v } 											\label{McircDef.eq}
	\end{align}
	where $M^{\circ 2}$ is the Hadamard (element-by-element) square of the matrix $M$, and we use 
	the notation $\Mb^\circ$ to denote the matrix of squared $H^2$ norms of subsystems of $\cM$
%
	\be
		\Mb^\circ ~:=~
		 \sum_{t=0}^\infty M^{\circ 2}_t  ~=~ \begin{bmatrix} \|\cM_{11}\|_2^2 & \cdots & \|\cM_{1n}\|_2^2	\\
													\vdots & & \vdots		\\
												\|\cM_{n1}\|_2^2 & \cdots & \|\cM_{nn}\|_2^2	
									\end{bmatrix} ,
	   \label{math2norm.eq}
	\ee	
	We therefore conclude that the non-zero eigenvalues of $\cL$ are precisely the eigenvalues of $\Mb^{\circ}$,
	and in particular, their spectral radii are equal. This is summarized
	in the following corollary.

	\begin{corollary} 				\label{uncor.cor} 
	For the uncertain system of Figure~\ref{M_Delta.fig} with uncorrelated uncertainties, the mean square 
	stability condition of Theorem~\ref{main_MIMO_mss.thm} becomes 
	\[
		\rho \left( \Mb^\circ  \right)  ~\leq~  1/ \gainb, 
	\]
	where $\gainb:=\expec{\gain_i^2}$ is the uncertainties' variance (assumed equal for all $i$) and 
	$\Mb^\circ $ is the matrix~\req{math2norm} of  squared $H^2$ norms of 
	$\cM$'s subsystems. 
	\end{corollary} 
	
\subsection{Repeated Perturbations}		 \label{repeated.subsec}
	
This case represents the opposite extreme to the uncorrelated perturbations case. Here we have all the perturbations
identical, i.e. 
\[
		\Gain(t) ~:=~ \Diag{  \rule{0em}{1em} \gain(t), ~\ldots,~ \gain(t) } ~=~ I~\gain(t), 
\]
where $\{\gain(t)\}$ is a scalar-valued iid random process and $I$ is the $n\times n$ identify matrix. In this case, 
all entries of the uncertainty correlation matrix are equal, i.e. 
$\expec{\gain_i(t) \gain_j(t)} = \gainb $, and therefore 
\[
	\Gainb ~=~ \gainb~ \bfo \bfo^* ,
\]
where $\bfo$ is the vector of all entries equal to $1$.  Now the Loop Gain Operator~\req{LGO} takes on a particularly 
simple form 
\begin{align*} 
				\cL (X) &= \big( \gainb~ \bfo \bfo^*  \big) \circ  \left( 	\sum_{t=0}^\infty M_{t} X M^*_{t} \right)
				= \gainb \sum_{t=0}^\infty M_{t} X M^*_{t}  						\\
				&= \gainb ~ \covl{\cM}{X}
\end{align*} 

The interpretation of $\cL$ in this case is simple. Referring to Equation~\req{covlim} we see that 
$\cL(X)$ is the steady-state covariance matrix of the output of the LTI system $\cM$ when its 
input is white noise 
 with covariance matrix $\gainb X$.  
 In particular, let the system $\cM$ have a state space realization $(A,B,C)$, then 
 for an eigen-matrix $X$ of $\cL$ with eigenvalue $\lambda$
 \be
  	\gainb ~ \covl{\cM}{X} = \cL(X)=\lambda X
  	~~ \Longleftrightarrow~~  \covl{\cM}{X} = \frac{\lambda}{\gainb} X, 
    \label{reppert.eq}
  \ee
  which implies that $X$ 
  satisfies the matrix equation 
 \begin{align*}
 	 Y ~-~ A YA^* ~=&  ~B XB^* , 		\\
	 	\frac{\lambda}{\gainb} X ~=& ~C Y C^*. 
 \end{align*} 
Equivalently, a single equation for $Y$ can be written
\be
	Y ~-~ A~Y~A^* ~=~  \frac{\gainb}{\lambda} ~ BC ~Y~ C^*B^* . 
  \label{lyaprep.eq}	
\ee
This is not a standard Lyapunov equation, but it can always be thought of as a generalized eigenvalue problem as follows. Define the linear operators $\cL_1$ and $\cL_2$ by 
\begin{align*} 
	\cL_1(Y) &:= ~ \gainb ~BC ~Y ~C^*B^* , \\
	\cL_2(Y) &:=~ Y ~-~A~Y~A' . 
\end{align*} 
Then equation~\req{lyaprep} can be rewritten as the generalized eigenvalue problem 
\be
	\cL_1(Y) ~=~ \lambda~ \cL_2(Y).
\ee 

Finally we note an interesting interpretation of covariance matrices that arise as eignmatrices of $\cL$ in the repeated 
perturbation case. As~\req{reppert} shows, these eigenmatrices are exactly those covariances of input processes to an LTI system $\cM$ with the property that the steady-state output covariance is a scalar multiple of the input covariance. These are very special processes, but there dynamical significance is not yet clear.

\subsection{Kronecker Product Representation of the General Case}		 \label{kronecker.subsec}
	
	For the general case of correlated uncertainties $\Gainb \neq I$, and it turns out that 
	 entries of the matrix~\req{math2norm} 
	alone are insufficient to characterize  $\rho\lb\cL\rb$ (and thus the MSS) condition. In the absence 
	of any spatial structure in  $\cM$ and $\Gainb$, one can always use a
	 Kronecker product representation of $\cL$. This  representation 
	gives some insight into the general case.

	
	Let $\vect{X}$ denote the ``vectorization'' operation of converting a matrix $X$ 
	into a vector by stacking up its columns. It is then standard to show that the 
	loop gain operator~\req{LGO} can  equivalently be written as 
	\begin{multline}
		\vect{\cL(X)} ~=~		\\
			\underbrace{
				\left(  \Diag{\rule{0em}{1em}\vect{\Gainb}} 
				~ \sum_{t=0}^\infty M(t) \otimes M(t)  \right) 	}_{\text{matrix representation of $\cL$}}  \vect{X} .
	   \label{kronrep.eq}
	\end{multline}
	Therefore, the eigenvalues (and corresponding eigenmatrices) of $\cL$ can be found by calculating 
	the eigenvalues/vectors of this $n^2\times n^2$ representation using standard matrix methods. 
	This is clearly not a desirable method even for only moderately large-scale systems. An alternative computational 
	method for large systems based on the power iteration is presented in Section~\ref{pi.sec}.

	The formula~\req{kronrep} however provides some insight into the general case when $\Gainb$ is not diagonal. In that 
	case, entries of the matrix representation will involve sums of the form 
	\be
		\sum_{t=0}^\infty M_{ij}(t) ~M_{kl}(t) . 
	   \label{inprod.eq}
	\ee
	These are inner products of impulse responses of different SISO subsystems of $\cM$. They can be thought 
	of as {\em autocorrelations} of the MIMO system $\cM$'s responses. 
	In the special case of uncorrelated uncertainties $\Gainb=I$, only terms for identical subsystems ( $(i,j)=(k,l)$ ) 
	appear, resulting in $H^2$ norms of subsystems. Thus it is seen that a condition involving only $H^2$ norms 
	like~\req{math2norm} is a highly special case. {\em  To characterize MSS in  correlated uncertainty cases, 
	one needs in addition other 
	system metrics, like the inner product between different subsystems' impulse responses~\req{inprod}.}

\section{Mean-square Performance} 					\label{perf.sec}

	The mean-square performance problem is a general formulation for linear systems with both additive and 
	multiplicative noise. It is straightforward to show that any Linear Time Invariant (LTI) system with both additive and 
	multiplicative noise can be redrawn in the form shown in Figure~\ref{perf.fig} ({\em Left}), where $\cM$ is LTI, 
	$w$ is the additive white noise, and 
	the multiplicative perturbations are grouped together in the diagonal matrix gain 
	$\Gain := {\rm diag}(\gain_1,\ldots,\gain_n)$. 		
	\begin{figure}[h]
		\begin{center}	
			\includegraphics[width=0.2\textwidth]{figures/G_delta_Perf.pdf}	~~~~~
			\includegraphics[width=0.2\textwidth]{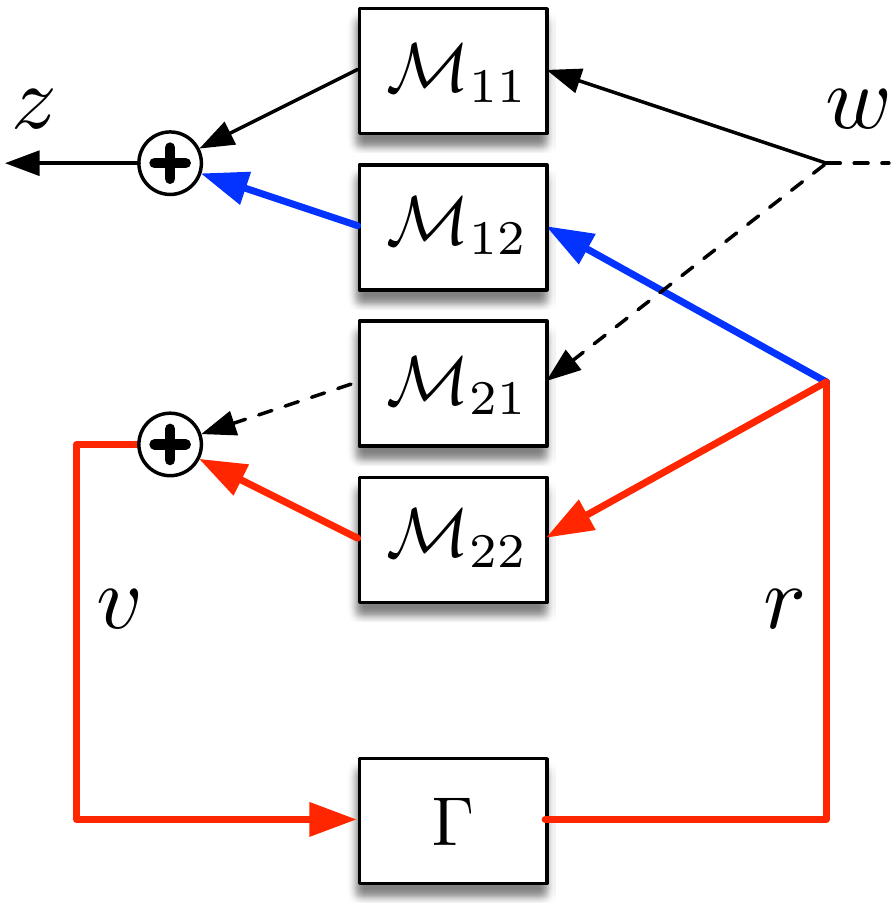}	
		\end{center}
		\caption{({\em Left}) The mean-square performance problem setting with additive noise $w$ as an 
		exogenous signal and multiplicative noise $\gain_i$'s as structured stochastic uncertainty. ({\em Right})
		Details of the various signal paths. The variance of $z$ is finite iff the feedback loop $(M_{22},\Gain)$ 
		is mean-square stable.}
	   \label{perf.fig}
	\end{figure}
	The assumption of whiteness of $w$ is made without any loss in generality. If the additive noise is colored, then 
	the coloring filter can be absorbed in the LTI dynamics of $\cM$ in the standard manner. 
		
	The mean-square performance problem is to find conditions for MSS stability of the feedback system, and to 
	 calculate the steady state covariance of the output $z$. It is clear 
	from Figure~\ref{perf.fig} ({\em Right}) that $z$ has finite covariance iff  the feedback subsystem $(\cM_{22},\Gain)$ 
	is MSS. The exact condition for this is that the spectral radius of the loop gain operator~\req{LGO} for $\cM_{22}$ 
	and $\Gain$
	\be
		\cL_{22} (X) ~:=~  \Gainb \circ  \left( 	\sum_{t=0}^\infty M_{22}(t) ~X~ M^*_{22}(t) \right) = \Gainb \circ \covl{\cM_{22}}{X}.
	   \label{cL22.eq}
	\ee
	has spectral radius less than 1. 
	
	The calculation of the covariance $\Zb$ proceeds
	similarly to Equation \req{sigmau} where the first steps are to relate the covariances of the signals in the lower 
	feedback loop. It is first noted that with assumption of MSS, all covariance sequences are bounded and have steady 
	state limits, so the following relations are written directly for those limits
	\begin{eqnarray*}
		\Rb
		& = & 	\Gainb \circ  \Vb			\\			
		& = &  \Gainb \circ  \lb   \sum_{t=0}^\infty
			{\begin{bmatrix} M_{21} & M_{22} \end{bmatrix}}_{t}  
			\begin{bmatrix} \Wb & 0 \\ 0 &   \Rb \end{bmatrix} 
			{\begin{bmatrix} M^*_{21} \\ M^*_{22} \end{bmatrix}}_{t}		\rb			\\
		& = & \Gainb \circ \lb      \sum_{t=0}^\infty		
			  M_{22}(t)    \Rb M^*_{22}(t) + M_{21}(t)  \Wb  M^*_{21}(t)  	\rb, 
	\end{eqnarray*}
	where for  simplicity, we have dropped the ``overbar'' notation~\req{bar_nota} 
	for the covariance limit (e.g. in this section, $\Rb$ stands
	for $\lim_{t\rightarrow\infty} \Rb_t$). 
	The expression for $\Vb$ 
	follows from \req{covlim} and the fact that both $w$ and $r$ are mutually uncorrelated and white. 
	The last equation can be rewritten in operator form using the definition~\req{cL22} and the notation of~\req{covlim}
	\begin{eqnarray*}
		\lb I ~-~ \cL_{22} \rb  \lb \Rb \rb & = &  \Gainb \circ  \covl{\cM_{21}}{\Wb}  			\\
		 \Rb  & = &  \lb I ~-~ \cL_{22} \rb^{-1}  \lb    \Gainb \circ  \covl{\cM_{21}}{ \Wb}    \rule{0em}{1em}	\rb
	\end{eqnarray*}
	
	Finally, to calculate the covariance of the output, note that 
	\begin{eqnarray}
		\Zb
		& = & \sum_{t=0}^\infty
			{\begin{bmatrix} M_{11} & M_{12} \end{bmatrix}}_{t}  
			\begin{bmatrix} \Wb & 0 \\ 0 & \Rb \end{bmatrix} 
			{\begin{bmatrix} M^*_{11} \\ M^*_{12} \end{bmatrix}}_{t} 		\nonumber \\	
		& = & \covl{\cM_{11}}{ \Wb } + 	\covl{\cM_{12}}{ \Rb }
																\nonumber	\\
		& = & \covl{\cM_{11}}{\Wb} + 	
			\covl{\cM_{12}}{
			 	\lb I - \cL_{22} \rb^{-1} \lb \Gainb \circ \covl{ \cM_{21}}{\Wb } \rule{0em}{1em} \rb 	}
			 											\nonumber
	\end{eqnarray}
	Note how this formula has a familiar feel to the Linear Fractional Transformations (LFT) from standard transfer 
	function block diagram manipulations. The difference being that these are operators on matrices rather than 
	vector signals. It is instructive to compare the above formula with that of
	the transfer functions operating on the 
	vector signals $w$ and $z$
	\[
		z ~=~ \lb \cM_{11} ~+~\cM_{12} \lb I - \Gain \cM_{22} \rb^{-1} \Gain \cM_{21}  \rb w  .
	\]
	This resemblance is more immediate in the SISO case (single SISO $\Gain$	and $w$ a scalar), where 
	the expression simplifies to 
	\[
		{\zb} ~=~ 
		\lb \|\cM_{11}\|_2^2 ~+~ \frac{\Gainb ~\|\cM_{12}\|_2^2 \|\cM_{12}]\|_2^2 }{1 ~-~ \Gainb \|\cM_{22}\|_2^2} \rb
		{\wb}
	\]
	
	The expression for $\Zb$ above is written to highlight the analogy with Linear Fractional Transformations LFT of 
	transfer functions. For computations, it is more convenient to write it in the following form of a system of two
	equations 
	\begin{align}
		\Rb - \Gainb \circ \covl{\cM_{22}}{\Rb} &=~ \Gainb \circ \covl{\cM_{21}}{\Wb} 		\label{Req.eq}\\ 
		\Zb &=~  \covl{\cM_{11}}{\Wb} + \covl{\cM_{12}}{\Rb} ,							\label{Zeq.eq}
	\end{align} 
	which indicates that in solving for $\Zb$, one has to go through the intermediate step of solving for $\Rb$ from~\req{Req}
	
\subsection{Uncorrelated  Uncertainties Case} 
	
	In this case we assume the uncertainties to have correlations 
	$\Gain=\gainb I$. The case when different uncertainties have different variances can be accommodated through 
	adding the appropriate multipliers in the system $\cM$. In this case it follows that the matrix $\Rb$ in~\req{Req} is 
	diagonal. If we assume in addition that $\Wb$ is diagonal, and we are only interested in the diagonal part of $\Zb$, 
	then equations~\req{Req}-\req{Zeq} can be rewritten in terms of matrix-vector multiplication using the 
	notation of~\req{McircDef}-\req{math2norm} where the vectors are the 
	diagonal entries of the respective covariance matrices
%
%
	\begin{multline*}
		\diag{\Rb} - \gainb~ \Mb_{22}^{\circ}~ \diag{\Rb} ~=~\gainb ~\Mb_{21}^{\circ} ~\diag{\Wb} 	\\ 
		\diag{\Zb} ~=~  \Mb_{11}^{\circ} ~\diag{\Wb} + \Mb_{12}^{\circ} ~\diag{\Rb} . 							
	\end{multline*} 
	Putting these two equations together by eliminating $\Rb$ gives 
	\[
		\diag{\Zb}  = 
			 \left( \Mb_{11}^{\circ} +	\gainb ~
			 \Mb_{12}^{\circ} 	\lb I - \gainb~ \Mb_{22}^{\circ } \rb^{-1}  \Mb_{21}^{\circ }  \right)   \diag{\Wb} 	.  	 
	\]
%
	Without loss of generality we can in addition assume 
	$\Wb=I$, for which  there is an even simpler expression for the total variance of $z$ 
	\be
			\tr{\Zb}  = 
			 \left| \Mb_{11}^{\circ } +	\gainb ~
			 \Mb_{12}^{\circ } 	\big( I - \gainb~ \Mb_{22}^{\circ } \big)^{-1}  \Mb_{21}^{\circ }  \right|	  	, 
	  \label{trZ.eq}
	\ee
	where $| M |$ stands for the sum of all the elements of a non-negative matrix $M$. 
	
	In the literature on robust stability analysis, there is often an equivalence between a robust performance condition 
	and a robust stability condition on an augmented system with an additional uncertainty. 
	The uncorrelated case here provides a version of such 
	a correspondence, and we will state it without loss of generality for the case of $\gainb=1$. 
	
	\begin{corollary} 
	Consider the  system of Figure~\ref{perf.fig} with uncorrelated uncertainties with variances $\expec{\gain_i^2(t)}=1$, and 
	scalar inputs and outputs $w$ and $z$ respectively. Then
	\[
		\frac{ \expec{z^2(t)} }{ \expec{ w^2(t)}} ~<~1 
		~~~~\Longleftrightarrow~~ ~~ 
		\rho \left(   \bbm \Mb_{11}^\circ  &  \Mb_{12}^\circ  \\  \Mb_{21}^\circ &   \Mb_{22}^\circ  \ebm   \right)  ~<~1.
	\]
	\end{corollary} 
	The proof is a simple application of the Schur complement on the $2\times 2$ block matrix, which implies that 
	the spectral radius condition is equivalent to the right hand side of~\req{trZ} being less than 1. 
	Note that the variance ratio condition is a performance condition, while the spectral radius condition is a MSS stability 
	condition for a system with an additional (fictitious) uncertainty in a feedback loop between $z$ and $w$. 

%
%
%

\section{State Space Methods and Computations} 					\label{ss.sec}

		While the input-output setting presented in this paper appears to be more expedient for analysis 
		and statement of results, it is often (though not always) the case that actual computations 
		are more conveniently carried out using state space representations. From one point of view, 
		this results in Linear Matrix Inequality (LMI) conditions. For large-scale system applications, we 
		write out a power iteration type algorithm that involves solving Lyapunov equations at each step of the 
		iteration. Finally we give a state space interpretation of the ``worst case covariance'' in the case where there 
		are no exogenous inputs. This turns out to be a worst case covariance of a random initial state.

	\subsection{Mean Square Stability Conditions} 				\label{pi.sec}
		
		Begin with the MSS problem of Theorem~\ref{main_MIMO_mss.thm}. 	
		Let the strictly-causal LTI system $\cM$ have the following realization
		\begin{eqnarray*}
		x_{t+1} & = &  A x_t + B u_t 	,		\\ 
		y_t  & = & C x_t , 
		\end{eqnarray*} 
		from which it follows that a corresponding realization for the covariance feedback system is
		\begin{eqnarray*}
            		\Xb_{t+1} & = & A \Xb_t A^* + B\Ub_t B^*			\\ 
            		\Yb_t  & = & C \Xb_t  C^* 						\\
            		\Ub_t & = & \Gainb \circ \Yb_t 	. 
		\end{eqnarray*}
		Taking the steady state limit produces the following representation of the loop gain 
		operator $\Rbb = \cL\lb\Ubb\rb$
		\begin{eqnarray}
            		\Xbb - A \Xbb A^*  & = &  B \Ubb B^*,			\label{lyapXbb.eq}		\\ 
            		 \Rbb & = & \Gainb \circ \lb C \Xbb  C^*  \rb 	.	\label{Rbb.eq}
		\end{eqnarray}
		Thus one method of computing the action of $\cL$ is to solve the Lyapunov equation~\req{lyapXbb} 
		for $\Xbb$ given the input covariance $\Ubb$, and then calculate $\Rbb$ from~\req{Rbb}. 	
		
		\subsubsection*{A Power Iteration Algorithm} 
		
		The above procedure for calculating the action of $\cL$ can now be used as follows in a
		power iteration 
		method for calculating $\rho\lb\cL\rb$ as recommended in~\cite{parrilo2000cone}. Starting from 
		an arbitrary initial matrix $\Pb_0\geq 0$ 
		\[
			\Pb_{k+1} ~=~ \cL\lb \Pb_{k} \rb ~/~ \|\Pb_k\| . 
		\]
		At each step the calculation of $\cL(P_k)$ involves equations~\req{lyapXbb}-\req{Rbb} as follows 
		\begin{eqnarray}
            		\Xb_{k+1} - A~ \Xb_{k+1}~ A^*  & = &  B \Pb_k B^*,			\label{PI_lyap.eq}		\\ 
            		 \Pb_{k+1} & = & \frac{1}{\|P_k\|} \lb \Gainb \circ \lb C \Xb_{k+1}  C^*  \rb  \rb	.	\nonumber
		\end{eqnarray}
		The major computational burden in each step is solving the Lyapunov equation~\req{PI_lyap}. 
		However, this power iteration algorithm is well-suited for use with sparse methods for solving 
		Lyapunov equations, which are themselves iterative procedures. 
		
		\subsubsection*{As an LMI}

		To calculate the spectral radius $\rho\lb\cL\rb$, one can set $\Rbb = \lambda \Ubb$ in the above 
		and find the largest real number $\lambda$ such that 		
		\begin{eqnarray*}
			\lambda \lb \rule{0em}{1em} \Xbb - A \Xbb A^*\rb    - 
					B \lb\rule{0em}{1em} \Gainb \circ \lb C \Xbb  C^* \rb \rb B^*  & = &  0 , 		\\
					\Xbb & \geq & 0 . 
		\end{eqnarray*}

	\subsection{Mean Square Performance} 
	
		For the mean square performance problem, begin with the following realization for the stable, 
		strictly causal system $\cM$ 
		\begin{eqnarray*}
		x_{t+1} & = &  A x_t + B_1 w_t +B_2 r_t 	,		\\ 
		z_t  & = & C_1 x_t ,							\\
		v_t & = & C_2 x_t  .
		\end{eqnarray*} 
		Since $w$ and $r$ are mutually uncorrelated and white, the
		 corresponding realization for the covariance feedback system is 
		\begin{eqnarray*}
            		\Xb_{t+1} & = & A \Xb_t A^* + B_1\Wb_t B_1^*	+ B_2\Rb_t B_2^*		\\ 
            		\Zb_t  & = & C_1 \Xb_t  C_1^* 						\\
            		\Vb_t  & = & C_2 \Xb_t  C_2^* 						\\
            		\Rb_t & = & \Gainb \circ \Vb_t 	. 
		\end{eqnarray*}
		The corresponding steady state equations are 
		\begin{align}
            		\Xbb - A \Xbb A^*  -   B_2  \lb \Gainb \circ \lb C_2 \Xbb  C_2^*  \rb \rb  B_2^*  & =   B_1\Wbb B_1^* ,  
																\label{msplyap.eq}\\ 
            		\Zbb  & =  C_1 \Xbb  C_1^* 		.
		\end{align}
		Therefore, the main step in evaluation the covariance of the output is solving the matrix
		equation~\req{msplyap} for the state covariance. This is not a standard Lyapunov equation, but 
		iterative algorithms, akin to those designed for large-scale Lyapunov equations, can be used 
		to tackle it.

		\subsection{Worst Case Covariances} 	
		
		The ``Perron eigen-matrix'' $\Ubwh$ of the loop gain operator $\cL$~(\ref{LGO.eq}) is by definition the matrix 
		that achieves the spectral radius of $\cL$, i.e. 
		\be
			\cL\big( \Ubwh \big) ~=~ \rho(\cL) ~ \Ubwh.
		  \label{Ubwh.eq}
		\ee
		In the necessity proof of Theorem~\ref{main_MIMO_mss.thm} (and the comment thereafter), it was shown that 
		this matrix has an interpretation as a sort of worst case covariance matrix. To recap, assume MSS is lost, so 
		$ \rho(\cL) >1$, and let the exogenous disturbances be such that $d=0$, and $w$ has covariance 
		$\expec{w_t w^*_t} = \Ubwh$. Then a consequence of inequality~\req{Usub_growth} is that
		 the covariance of the signal $u$ will grow at a geometric rate of 
		 \[
		 	\expec{u_tu_t^*} ~\geq~ c ~\alpha^t ~\Ubwh ,
		 \]
		where for any $\epsilon>0$, we can choose $\alpha = \rho(\cL)-\epsilon$, and $c>0$
		is some constant. 
		
		An alternative interpretation which does not require exogenous inputs can also be given. In this scenario, 
		the exogenous inputs $w$ and $d$ are set to zero, but the system $\cM$ has some nonzero random initial 
		state $x_0$ with covariance $\expec{x_0x_0^*}=:\Xb_0$.  In this case the evolution of the state covariance
		has the following dynamics 
		\begin{equation} \label{ssIC.eqn}
			\begin{aligned}
			\Xb_{t+1} &=~ A\Xb_t A^* + B\Ub_t B^* ,  ~~~~ \Xb_0 = \expec{x_0x_0^*} ,	\\ 
			\Ub_t &=~ \Gainb \circ \left( C \Xb_t C^* \right) .
			\end{aligned}
		\end{equation}
		Now let $\Ubwh$ be an eigenmatrix of $\cL$ as (\ref{Ubwh.eq}). Then
		\begin{align*}
		\Ubwh &= \frac{1}{\rho(\cL)}\cL\big(\Ubwh \big)  = \frac{1}{\rho(\cL)}  \Gainb \circ \left( \sum_{\tau=0}^{\infty} M_t \Ubwh M_t^*\right) \\
		&=\frac{1}{\rho(\cL)} \Gainb \circ  \left(C \sum_{t=0}^\infty A^t B \Ubwh B^* {A^*}^t C^* \right) \\
		&=: \Gainb \circ \left( C \hat{\mathbf X} C^* \right), 
		\end{align*}
		where $\hat{\mathbf X} := \frac{1}{\rho(\cL)} \sum\limits_{t=0}^\infty A^t B \Ubwh B^* {A^*}^t$ is the worst-case covariance of the state. It can be calculated from $\Ubwh$ using the following algebraic Lyapunov equation
		$$ \hat{\mathbf X} - A \hat{\mathbf X} A^* = \frac{1}{\rho(\cL)} B \Ubwh B^*.$$
		Note that setting $\bf X_0 = \hat{\mathbf X}$ yields $\bf U_0 = \Ubwh$. By substituting $\Wb_t = 0$ in (\ref{recursion0.eqn}) and carrying out the same argument in the necessity proof of Theorem~\ref{main_MIMO_mss.thm}, we obtain
		\begin{equation*}
			\Ub_{Tk} \geq \cL_T^k(\Ubwh) \geq (\rho(\cL) - \epsilon c)^k \Ubwh =: \alpha^k \Ubwh.
		\end{equation*}
		This calculation shows that $\{\Ub_{Tk}\}$ is a geometrically growing sequence since $\epsilon$ can be chosen small enough so that $\alpha >1$. Consequently, by (\ref{ssIC.eqn}), we have 
		$$\Xb_{Tk+1} = A \Xb_{Tk} A^* + B \Ub_{Tk} B^* \geq \alpha^k B\Ubwh B^*,$$
		and therefore $\{\Xb_{Tk+1}\}$ is also a geometrically growing sequence.

\section{Conclusions and Discussion} 								\label{conc.sec}

	In this paper, we have study the Mean-Square Stability (MSS) and performance of linear time-invariant systems in feedback with stochastic disturbances. We derive the necessary and sufficient conditions of MSS by adopting a purely input/output approach, and thus state space realizations are treated as a special case. Our treatment leads to uncover a linear operator whose (1) spectral radius fully characterizes the conditions of MSS, and whose (2) ``Perron-Frobenius Eigenmatrix" characterizes the fastest growing modes of the covariances when MSS is lost.
	
	This paper treats the discrete-time setting where the stochastic disturbances are all white in time but are allowed to have ``spatial correlations". Future work in this line of research includes addressing the continuous-time setting and generalizing the analysis for stochastic disturbances that are correlated in time as well.

\bibliographystyle{ieeetr}
\bibliography{Master}

\appendix

\subsection{Independence} 								\label{indep.appen}

	For any two independent  random variables $a$ and $b$, $\expec{ab} = \expec{a} \expec{b}$. 
	Let $X$, $Y$ and $Z$ be (possibly matrix-valued) random variables. Assume $Y$ is independent of $X$ 
	and $Z$. Then 
	\begin{eqnarray}
		\expec{ XYZ} 
		& = & \expec{ \sum_{j,k} x_{ij} y_{jk} z_{kl}  } ~=~   \sum_{j,k} \expec{ x_{ij} y_{jk} z_{kl}  }	\nonumber	\\
		& = & 
		 \sum_{j,k} \expec{ x_{ij}  z_{kl}  }	~\expec{y_{jk}} 	\nonumber		\\
		 & = & 
		  \sum_{j,k} \expec{ x_{ij} ~ \expec{y_{jk}} ~ z_{kl}   \rule{0em}{1em} }	 	\nonumber		\\
		 & = & 
		\expec{ X~ \expec{Y} ~Z  \rule{0em}{1em}}							\label{mat_indep.eq}
	\end{eqnarray}

\subsection{Monotone Systems and Signals}		\label{monotone.appen}	

	For linear monotone systems, it is first shown that positive feedback interconnections are also monotone. It is 
	then shown that a time-invariant monotone system preserves signals' temporal order. 
	
	Begin with general comments about causal discrete-time systems. Signals are identified 
	with $\ell$, the set of vector-valued
	sequences over $\Z^+$. A causal linear system $\cG:\ell \rightarrow \ell$ is a mapping on $\ell$, and it can be 
	identified with a lower-triangular semi-infinite matrix. A (possibly unstable) positive feedback system is {\em well posed} 
	if $(I-\cG)^{-1}:\ell \rightarrow \ell$ is well defined. 
	
	Let $P_T:\ell \rightarrow \ell_T$ be the ``projection'' 
	\[
		(P_T f)(t) ~:=~   f(t) ,~~~  t\leq T, 
	\]
	where $\ell_T:= \left\{ f:\{0,\ldots,T\} \rightarrow \R^n \right\}$ is the space of finite sequences of length $T+1$. 
	With a slight abuse of notation, define the ``injection'' $P^\dagger_T:\ell_T \rightarrow \ell$ by 
	\[
		\left( P^\dagger_T f \right)(t) ~:=~ 
			\left\{  \begin{array}{lcl} 	f(t) & & t\leq T \\ 0 & & t>T	\end{array} \right. 		. 
	\]
	Clearly $P_T P^\dagger_T = I$, and for any system $\cG$, $P_t \cG P^\dagger_T$ is the finite matrix
	 ``upper left block'' of 
	its semi-infinite matrix representation. It can be thought of as a finite time-horizon restriction of $\cG$. 
	Causality of $\cG$ implies that 
	\[
		P_T \cG^n P^\dagger_T ~=~ \lb P_T \cG P^\dagger_T \rb^n 
	\]
	for any power $n$, and if $\cG^{-1}$ exists, then 
	\[
		P_T \cG^{-1} P^\dagger_T ~=~ \lb P_T \cG P^\dagger_T \rb^{-1}.
	\]
	We finally note that while $P_T^\dagger P_T \neq I$, for any  causal system
	 $\cG$, and any time $T$ we have 
	\[
		P_T~ \cG~ P_T^\dagger  P_T
		~=~ 
		P_T~ \cG
	\]
	
	\subsubsection{Feedback Interconnections}
	
	\begin{theorem} 				\label{monfdbk.thm}
		The sum, cascade and positive feedback interconnections of causal
		monotone linear systems are monotone. 
	\end{theorem}
	\begin{proof}	
	Closure under sums and cascades is obvious from the definition~\req{mondef}. This in particular implies that 
	powers $\cM^n$ of any monotone system $\cM$ are also monotone. Therefore, if the Neuman series 
	\[
		\lb I - \cM \rb^{-1} ~=~ \sum_{n=0}^\infty \cM^n 
	\]
	can be shown to converge in an appropriate sense, then positive feedback interconnections are also monotone. 
	A convergence argument is now given. It is similar to successive iteration 
	schemes for Volterra operators~\cite{rienag90} (see also~\cite[Appendix]{bamieh2005convex}). 
	
	Consider the partial series product
	\begin{multline}
		P_T(I-\cM) P^\dagger_T P_T   \lb \sum_{n=0}^N \cM^n \rb P^\dagger_T 	\\
		= P_T(I-\cM)  \lb \sum_{n=0}^N \cM^n \rb P^\dagger_T 		
		 =   I- P_T \cM^{N+1} P^\dagger_T				\\
		= I -  \lb P_T \cM P^\dagger_T\rb^{N+1},
	\end{multline}
	where the last equality follows from the causality of $\cM$. Strict causality of $\cM$ means
	$P_T \cM P^\dagger_T$ is 
	just a $(T+1)\times (T+1)$ strictly lower-triangular matrix. It is therefore nilpotent and 
	\[
		\lb P_T \cM P^\dagger_T\rb^{N+1} = 0 , ~~N\geq T. 
	\]
	The conclusion is then that 
	\[
		 P_T(I-\cM)^{-1} P^\dagger_T  =
		\lb P_T(I-\cM) P^\dagger_T \rb^{-1} \!\!\! \!=  P_T \lb \sum_{n=0}^T \cM^n \rb P^\dagger_T .
	\]
	This means that for each $T$, the finite horizon restriction $ P_T(I-\cM)^{-1} P^\dagger_T $ is monotone, and 
	therefore the system $(I-\cM)^{-1}$ itself must be  monotone. 
	\end{proof}
	
	\subsubsection{Preserving Monotonicity of Signals}				\label{tempord.appen}

            	For time invariant systems, the above definition of monotonicity has an additional implication in that
            	non-decreasing input sequences produce non-decreasing output sequences. 
            	Consider the input-output pair $\Yb = \cM ( \Ub)$. Let
            	$\cS$ be the right shift operator on sequences
            	\[
            		\left( \cS  \Ub  \right) (t)  
            		~:=~ \left\{  \begin{array}{lcl}  \Ub(t-1), & & t\geq 1,\\  0, & & t=0. 	\end{array}  \right. 
            	\]
            	The time invariance of $\cM$ means that $\cM \cS^n = \cS^n \cM$ for all powers $n\geq 1$. 
            	Recall that a signal is said to be {\em monotone} (or {\em non-decreasing}) if
            	\[
            		t_1\leq t_2  ~~~\Longrightarrow~~~ \Ub(t_1) \leq \Ub(t_2).  
            	\]
            	An equivalent condition for a signal to be monotone is 
            	\[
            	 	\forall~ n\geq 1, ~~~~	\cS^n \Ub ~\leq~ \Ub, 
            	\]
            	where the relation $\leq$ is the pointwise ordering on signals~\req{pointwise_ordering}. Now calculate that 
            	for any $n\geq 1$
            	\[
            		\cS^n \Yb ~=~ \cS^n \cM(\Ub) ~=~ \cM ( \cS^n \Ub) ~\leq~ \cM(\Ub) ~=~ \Yb, 
            	\]
            	where the inequality follows from $\cM$ being monotone together with $\Ub$ being  nondecreasing. 
            	One can therefore conclude that a {\em time-invariant monotone system preserves monotonicity 
            	of signals}.

\subsection{Proof of Lemma~\ref{pastpresent.lemma} }

	First observe that the mappings from the exogenous inputs to all signals in the loop are 
	\be
	\bbm 
	u \\ y \\ v \\ r 
	\ebm 
	=
	\bbm 
	(I-\Gain\cM)^{-1} 			&	\Gain (I-\cM\Gain)^{-1}  \\ 
	(I-\cM\Gain)^{-1}\cM		&	\cM \Gain (I-\cM\Gain)^{-1}\\
	(I-\cM\Gain)^{-1}\cM		&	 (I-\cM\Gain)^{-1} \\
	\Gain(I-\cM\Gain)^{-1}\cM 	&	\Gain (I-\cM\Gain)^{-1}  
	\ebm 
	\bbm 
	w \\ d 
	\ebm	 . 
	\label{GangOfFour.eq}	
	\ee
	Thus one needs to investigate the causal dependencies of all thethe four mappings in the matrix of (\ref{GangOfFour.eq}). 
	Consider first the mapping $\big( I - \cM  \Gain \big)^{-1}$.
	Over the time horizon $[0,t]$ this operator  can be written in partitioned matrix form as 
	\[
		 \bbm  I 				 				& 			& & \\ 
		 	    -M_1 \Gain_0 					& 	I		& & \\ 
			   								& 		\ddots	& \ddots & \\ 
			-M_t \Gain_0 			& \cdots 		& -M_1\Gain_{t-1} & I
		\ebm	^{-1} 
	=
		\bbm	I & & \\ 
				 & \ddots & \\ 
				* &  & I
		\ebm		
	\]
	Note the strictly block lower-triangular structure of $\cM\Gain$ which is a consequence of the strict causality 
	of $\cM$. The $*$ blocks are functions of $\Gain_0, \ldots, \Gain_{t-1}$, and are independent of $\Gain_\tau$, $\tau\geq t$. 
	
	Using this we write the equation for $u$ over the time horizon 
	\begin{align*}
		\bbm 
			u_0  \\ \vdots  \\ u_t    
		\ebm
	~=& 
		\bbm	I & & \\ 
				 & \ddots & \\ 
				* &  & I
		\ebm		
		\bbm 
			w_0  \\ \vdots  \\ w_t    
		\ebm								\\
	& +
		\bbm	\Gain_0 & & \\ 
				 & \ddots & \\ 
				 &  & \Gain_t
		\ebm		
		\bbm	I & & \\ 
				 & \ddots & \\ 
				* &  & I
		\ebm		
		\bbm 
			d_0  \\ \vdots  \\ d_t    
		\ebm
	\end{align*}
	Now recall that all noise terms $\{\Gain_t\}$, $\{d_t\}$ and $\{w_t\}$ are assumed mutually independent. The
	equation above shows that $\{u_0,\ldots,u_t\}$ are a function of only past and present values of 
	$\{\Gain_t\}$, $\{d_t\}$ and $\{w_t\}$. 
	
	It is now clear that by repeating the above argument for each of the signals that whenever $\Gain$ is preceded by the operator $\cM$, the dependence on the present value of $\Gain$ is killed by the strict causality of $\mathcal M$. Therefore, the following conclusion can be 
	stated: for any signal  in~\req{GangOfFour} that involves $\Gain$ only when it is preceded by the mapping $\cM$ (namely the signals $y$ and $v$),  present values of that signal are independent 
	of current and future values of $\Gain$.
    Otherwise (namely the signals $u$ and $r$),  their present values are independent of future values of $\Gain$ only.

\subsection{Some Properties of the Hadamard Product}


	Let $\pi$ be a permutation matrix, this means each row and each column contains exactly one non-zero element equal to 1. A non-zero element in location $ij$ implies that the $j$'th component of a vector is mapped to the $i$'th component of the vector. Thinking of the inverse operation, clearly $\pi^{-1} = \pi^*$. There are in general no simple relations between 
	the regular matrix product and the Hadamard product. However, for permutation matrices, we have the simple 
	relation 
	\be 
		\pi_1 \left( A \circ B \right) \pi_2 ~=~ \left(  \pi_1A \pi_2\right) \circ \left(   \pi_1 B \pi_2 \right), 
	   \label{had_prop.eq}	
	\ee
	which is obviously true since for any matrix $M$, the matrix $\pi_1 M \pi_2$ is simply a rearrangement 
	of  its entries. 
%
%

\subsection{Proof of Necessity in Lemma~\ref{siso.lemma}}
	\label{nec1.appen}

	\textit{``only if''}) To simplify notation, assume $\Gainb = 1$. The general case follows by scaling.   
	It will be shown next that if $ \|\cM\|_2^2 \geq 1$, 
	 $w$ is a white, constant variance 
	process and $d=0$,  then $\ub$ is an  unbounded sequence. 
	
	From~\req{rho_uk}, \req{siso_delta_variances}, \req{sume} and~\req{siso_variances} the sequence $\ub$ satisfies the 
	following recursion  
	\be
		\ub_t ~ = ~  \yb_t ~+~ \wb_t
		~ = ~
				 \sum_{\tau=0}^t M_{t-\tau}^2 ~\ub_\tau~+~  \wb_t  .
	   \label{rhorec.eq}
	\ee
	This recursion may not be of finite order (e.g. if $\{M_t\}$ is not FIR) and it is therefore not clear how 
	to use it to estimate the growth of $\ub$. However, it can be replaced with a simple recursive {\em inequality} 
	for a subsequence of $\ub$, for which a growth estimate is immediately obtained. This is the essence of the 
	remainder of the proof. 
		
	Note that the quantity 
	$
		 \alpha~:=~  \sum_{\tau=0}^{T}   M_\tau^2  
	$
	can be made arbitrarily close to $\|\cM\|^2_2 \geq 1$ by choosing the time horizon $T$ sufficiently large. 
	It will now be shown that the subsequence $\{ \ub_{T k}; ~k\in\Z^+ \}$ is unbounded. 
	First, the non-negativity of all sequences in~\req{rhorec} gives a recursive {\em inequality} 
	for the subsequence  $\{ \ub_{T k} \}$ 
	\begin{eqnarray*} 
		\ub_{T k} 
			& = &  \sum_{\tau=0}^{T k}  M_{T k-\tau}^2~ \ub_\tau ~+~ \wb_{Tk} \\
			& \geq & \sum_{\tau=T(k-1)}^{T k}  M_{T k-\tau}^2~ \ub_\tau ~+~ \wb_{Tk}  \\
			& \geq & \left( \sum_{\tau=0}^T M_\tau^2 \right)  \min_{T (k-1) \leq \tau \leq T k} \ub_\tau 
							~+~ \wb_{Tk} 			\\
			& = &     \alpha ~\ub_{T(k-1)} ~+~ \wb_{Tk} , 
	\end{eqnarray*} 
	where the last equality follows from the monotonicity of the sequence $\ub$. 
	The above is a difference inequality which has the initial condition 
	$\ub_0 =  \rbo_0+ \wb_0 = \wb_0$ 
	($\rbo_0=0$ since $d=0$ and $\cM$ is strictly causal).
	A simple induction argument gives  
	\be
		\ub_{T k} ~\geq~ 	 \sum_{r=0}^k \alpha^r ~\wb_{T(k-r)} 
	   \label{alphabound.eq}
	\ee
	which is a convolution of $\{\alpha^k\}$ with the subsequence $\{\wb_{Tk}\}$. 
	Now if $\|\cM\|_2>1$, then a time horizon $T$ can be chosen such that $\sum_{\tau=0}^{T}   M_\tau^2  =: \alpha >1$.  
	The monotonicity of the sequence $\db$ 
	and~\req{alphabound}  implies that $\{\ub_{T k}\}$ (and thus $\ub$) is a 
	geometrically increasing sequence. 
	
	The case $\|\cM\|_2 = 1$ is slightly more delicate. 
	 $T$ can be chosen such that $\alpha$ is as close to $1$ as desired. For $\alpha<1$ one also 
	has
	\[
		\lim_{k\rightarrow\infty} \left( \alpha^k+ \cdots + \alpha + 1 \right) ~ =~ \frac{1}{1-\alpha}.
	\]
	For any $\epsilon>0$,  $k$ can be 
	chosen such that\footnote{$\db_0$ is chosen as a simple lower bound on the entire sequence $\db$. Other choices can 
		produce better lower bounds on $\ub$. } 
	\[
		\ub_{T k} ~\geq~ \frac{\wb_0}{1-\alpha} - \epsilon .
	\]
	Now given any lower bound $B$, choose $T$ and $k$ such that 
	$\alpha$ is sufficiently close to $1$ and $\epsilon$ is sufficiently small so that 
	\[
		\ub_{T t} ~\geq ~  \frac{\wb_0}{1-\alpha} - \epsilon ~>~ B. 
	\]
	This proves that $\ub$ is an unbounded sequence even though it may not have geometric
	growth in the case $\|\cM\|_2 = 1$. 

\subsection{Lemmas Used in the Proof of Necessity in Theorem~\ref{main_MIMO_mss.thm} }
\label{.appen}
Throughout this appendix, let $\emax{A}, \emin{A}$ and $||A||$ denote the largest eigenvalue, smallest eigenvalue, and the spectral norm of any matrix $A$, respectively. As the rest of the paper, matrix inequalities are understood as semi-definite ordering on matrices. Note that if $A \geq 0$, then $||A|| = \emax{A}$. Furthermore, let $\mathcal N(A)$ and $\mathcal R(A)$ denote the null space and range space of $A$, respectively. We now present two lemmas that are required for the proof of necessity for Theorem~\ref{main_MIMO_mss.thm}.
\begin{lemma} \label{nullspace.lemma}
	Let $A, B \in \mathbb R^{n\times n}$, such that $0 \leq B \leq A$. Then $\mathcal N(A) \subseteq \mathcal N(B)$.
\end{lemma}
\begin{proof}
	Since $0 \leq B \leq A$, then $\forall v \in \mathbb R^n$ we have ${0 \leq v^* B v \leq v^* A v}$. Particularly, let $v \in \mathcal N(A)$, then ${v^*Av = 0}$ which implies that $v^*Bv = 0$ as well. We are now left with proving that $Bv=0$. 
	
	Let $B = U\Sigma U^*$ be the eigendecomposition of $B$, then ${v^* U \Sigma U^* v = 0}$. Setting $w := U^*v$ yields $w^* \Sigma w = 0$ which implies that $\Sigma w = 0$ (because $\Sigma$ is diagonal with nonnegative entries). Finally, we have	${Bv = U \Sigma U^* v = U \Sigma w = 0}$, which completes the proof.
\end{proof}
\begin{lemma} \label{limit_ordering.lemma}
	Let $A, B \in \mathbb R^{n\times n}$ and $\rho, \epsilon > 0$ such that ${0 \leq B \leq \rho A}$ and $||\rho A - B|| \leq \epsilon ||A||$.
	Then $\exists c >0$ such that $ B \geq (\rho - \epsilon c) A$. 
\end{lemma}
\begin{proof}
	The proof is carried out for the case where $A > 0$ first. Then the result is exploited to prove the more general case where $A \geq 0$. 
	
	``$A > 0$") $\forall v \in \mathbb R^n$, we have
	\begin{align*}
	v^* (\rho A - B) v &\leq ||\rho A - B ||~||v||^2 \leq \epsilon ||A||~||v||^2, 
	\end{align*}
	where the first inequality follows by noting that $\rho A - B \geq 0$ and thus $\emax{\rho A - B}=||\rho A - B|| $. Recalling that ${v^*Av \geq \emin{A} ||v||^2}$ and $A > 0$ (i.e. $\emin{A} > 0$), we obtain the following upper bound
	$$ \epsilon ||A||~||v||^2 \leq \epsilon ||A|| \frac{1}{\emin{A}} v^*Av =: \epsilon c~v^*Av,$$
	where $c:= \emax{A} / \emin{A}$ is the condition number of $A$. Now, $\forall v \in \mathbb R^n, v^*(\rho A - B)v \leq \epsilon c~v^*Av$ which implies that $\rho A - B \leq \epsilon c A$. Finally, rearranging the last inequality completes the proof.
	
	``$A\geq 0$") Let $r < n$ denote the rank of $A$ so that its eigendecomposition can be written as
	\begin{equation*}
	A = U \Sigma U^* = 
	\matbegin
	\begin{array}{c:c}
	U_1 & U_2 \\ \\ 
	\end{array}
	\matend
	\matbegin
	\begin{array}{c:c}
	\Sigma_r & 0\\ \hdashline 0 & 0
	\end{array}
	\matend 
	\matbegin
	\begin{array}{cc}
	U_1^* & \\ \hdashline U_2^* & 
	\end{array}
	\matend, 
	\end{equation*}
	where $U$ is a unitary matrix and $\Sigma_r$ is a diagonal matrix with strictly positive entries. Before we continue the proof, observe that this matrix partitioning indicates that 
	\begin{itemize}
		\item $\mathcal N(A) = \mathcal R (U_2)$, and thus $AU_2 = 0$.
		\item $A = U_1 \Sigma_r U_1^*$ and thus $U^*_1 A U_1 > 0$ (since $U_1^*U_1 = I$).
		\item Lemma~\ref{nullspace.lemma} guarantees that $\mathcal N(A) \subseteq \mathcal N(B)$, and thus $AU_2 = BU_2 = 0$.
	\end{itemize} 
	Multiplying all sides of the inequality $0 \leq B \leq \rho A$ by $U^*$ from the left and $U$ from the right preserves its ordering, then $0 \leq U^* B U \leq ~\rho U^* A U$ which implies
	\begin{align*}
	0 \leq 
	\matbegin
	\begin{array}{c:c}
	U_1^* B U_1 & 0\\ \hdashline 0 & 0
	\end{array}
	\matend \leq 
	\matbegin
	\begin{array}{c:c}
	\rho ~U_1^* A U_1 & 0\\ \hdashline 0 & 0
	\end{array}
	\matend.
	\end{align*}
	This is a consequence of $AU_2 = BU_2 = 0$ and $U_2^*A = U_2^*B = 0$. Define $A_{11} := U_1^*A U_1$ and $B_{11} := U_1^*BU_1$, then we have 
	\begin{equation} \label{cond1.eqn}
	0 \leq B_{11} \leq \rho A_{11}.
	\end{equation}
	Furthermore, recalling that $||\rho A - B || \leq \epsilon ||A||$, and knowing that the spectral norm of a matrix is preserved under multiplications by unitary matrices, we obtain ${||U^*(\rho A - B) U || \leq \epsilon || U^* A U||}$ which implies
	\begin{equation} \label{cond2.eqn}
	||\rho A_{11} - B_{11}|| \leq \epsilon || A_{11}||.
	\end{equation}
	Since $A_{11} > 0$, the first part of the proof (``$A > 0$") can be invoked here by exploiting (\ref{cond1.eqn}) and (\ref{cond2.eqn}) to obtain ${B_{11} \geq (\rho - \epsilon c) A_{11}}$ where $c:= \emax{A_{11}}/\emin{A_{11}}$. This implies that
	\begin{equation*}
	\matbegin
	\begin{array}{c:c}
	U_1^* B U_1 & 0\\ \hdashline 0 & 0
	\end{array}
	\matend \geq (\rho - \epsilon c)
	\matbegin 
	\begin{array}{c:c}
	U_1^* A U_1 & 0\\ \hdashline 0 & 0
	\end{array}
	\matend.
	\end{equation*}
	Finally, multiplying both sides of the inequality by $U$ from the left and $U^*$ from the right completes the proof because
	\begin{equation*}
	U^*AU = 
	\matbegin 
	\begin{array}{c:c}
	U_1^* A U_1 & 0\\ \hdashline 0 & 0
	\end{array}
	\matend \implies
	U
	\matbegin 
	\begin{array}{c:c}
	U_1^* A U_1 & 0\\ \hdashline 0 & 0
	\end{array}
	\matend
	U^* = A, 
	\end{equation*}
	and the same reasoning holds for $B$.
	
\end{proof}

\end{document}